\documentclass[a4paper, onecolumn, unpublished]{quantumarticle}
\pdfoutput=1

\usepackage[utf8]{inputenc}
\usepackage[english]{babel}
\usepackage[T1]{fontenc}
\usepackage{amsmath, amsthm}
\usepackage[numbers]{natbib}
\usepackage{hyperref}
\usepackage{amssymb, physics}
\usepackage{graphicx}
\usepackage{nameref} %
\usepackage{tikz}
\usetikzlibrary{decorations.pathreplacing, arrows.meta,positioning,fit,backgrounds,calc}
\usepackage{booktabs}
\usepackage{adjustbox}
\usepackage{subcaption}
\usepackage{fontawesome5}
\usepackage{xcolor}
\usepackage{placeins}
\usepackage{quantikz}

\newcommand{\clifftplay}[1]{\href{https://unitaryfoundation.github.io/clifft/playground/?url=https://raw.githubusercontent.com/unitaryfoundation/clifft-paper/main/qec_bench/circuits/#1}{\faPlay}}

\newcommand{\clifftplaydefault}[1]{\href{https://unitaryfoundation.github.io/clifft/playground}{\faPlay}}

\definecolor{UFYellow}{HTML}{FFF200}
\definecolor{UFBlack}{HTML}{111111}
\definecolor{UFGray}{HTML}{595959}
\definecolor{UFBlue}{HTML}{1F5FA6}
\definecolor{UFOrange}{HTML}{D97706}
\definecolor{UFRed}{HTML}{B42318}

\newtheorem{theorem}{Theorem}
\newtheorem{lemma}{Lemma}
\newtheorem{definition}{Definition}

\makeatletter
\newcommand{\myfirstpagefootnote}[1]{%
  \insert\footins{\footnotesize #1\par}%
}
\apptocmd{\@printauthors}{%
  \myfirstpagefootnote{GitHub repository: \url{https://github.com/unitaryfoundation/clifft}}}{}{}
\makeatother

\begin{document}

    \title{Clifft: Fast Exact Simulation of Near-Clifford Quantum Circuits}

    \author{Bradley Chase}
    \email{brad@unitary.foundation}
    \author{Farrokh Labib}
    \email{farrokh@unitary.foundation }

    \affiliation{Unitary Foundation}

    \maketitle

    \begin{abstract}
    Exact classical simulation of fault-tolerant quantum circuits remains limited by a tradeoff between exponential state vector scaling, exponential $T$-count scaling in stabilizer-rank approaches, and per-shot tracking overhead in sparse generalized stabilizer simulators. In this work, we introduce Clifft, an open-source simulator that shifts the dominant exponential cost from the total qubit count to a dynamic active subspace by factoring the quantum state into an offline Clifford frame, an online Pauli frame, and a dynamically sized active state vector. This architecture resolves deterministic Clifford coordinate transformations ahead of time, generalizing Stim's compile-once, sample-many execution model to circuits with non-Clifford operations. Consequently, exponential simulation costs are determined by the peak active virtual dimension, which expands during non-Clifford operations and contracts during measurements. Clifft remains within a constant factor of standard tools in the pure-Clifford and non-Clifford limits, while delivering up to orders-of-magnitude throughput gains over GPU-accelerated near-Clifford simulators on low-magic fault-tolerant benchmarks. Executing on commodity CPUs and exposing a Stim-like API, Clifft enables, to our knowledge, the first exact end-to-end simulation of magic state cultivation including the escape stage, over hundreds of billions of shots. These simulations show that escape-stage failures suppress the discrepancy between the true $T$-gate circuit and its $S$-proxy at low decoder-gap thresholds, while at high thresholds the full-protocol behavior approaches the larger discrepancy observed in the cultivation stages alone.
    \end{abstract}

    \section{Introduction}

    Classical simulation plays a central role in the development of fault-tolerant quantum computing. It is used to design and validate fault-tolerant protocols, estimate logical error rates, study decoder behavior, and compare architectural tradeoffs before hardware can reach the relevant scale~\cite{bluvsteinFaulttolerantNeutralatomArchitecture2026, google2023suppressing, fowler2012surface, Dennis2002Topological, Higgott2022PyMatching, Bravyi2024}. For circuits composed entirely of Clifford operations and Pauli measurements, this workflow is enabled by high-throughput stabilizer simulators such as Stim~\cite{gidneyStimFastStabilizer2021}. By combining efficient tableau methods of Aaronson and Gottesman \cite{aaronsonImprovedSimulationStabilizer2004} with a compile-once, sample-many Pauli-frame workflow and highly optimized kernels, Stim made it practical to sample the billions and trillions of shots needed to evaluate the performance of protocols with increasingly lower logical error rates. Its wide adoption illustrates how rapidly the field can advance when provided with a purpose-built simulator capable of operating at scale.

    However, many important fault-tolerant protocols are not purely Clifford \cite{gidneyMagicStateCultivation2024,itogawaEfficientMagicState2025, chamberlandVeryLowOverhead2020, litinskiMagicStateDistillation2019}. Universal quantum computation requires non-Clifford resources~\cite{Bravyi2005Universal,howardContextualitySuppliesMagic2014}, and the corresponding circuits quickly leave the regime where stabilizer methods apply directly. Existing simulation strategies each address part of this problem, but with different tradeoffs. Dense state vector methods are exact and flexible, but their memory cost grows exponentially with the total qubit count \cite{qiskit2024, quantum_ai_team_and_collaborators_2025_4067237}. Tensor-network and matrix-product-state methods can be effective when entanglement remains limited, but the deep, highly connected operations of some fault-tolerant gadgets incur heavy costs~\cite{aryalDensityMatrixRenormalizationGroup2023, huangClassicalSimulationQuantum2020, markov2008simulating, Schuch2008}. Low-magic simulation methods based on stabilizer-rank decompositions, Pauli-propagation and generalized stabilizer approaches, instead try to exploit the fact that many fault-tolerant circuits remain close to the Clifford regime~\cite{bravyiImprovedClassicalSimulation2016, rallSimulationQubitQuantum2019, kissingerClassicalSimulationQuantum2022, sutcliffeFastClassicalSimulation2025, haenelTsimFastUniversal2026, surtiEfficientSimulationLogical2025}. A complementary line of work exploits continuous algebraic structure to obtain efficient simulation when the circuit is constrained enough, including matchgate and fermionic-linear-optics circuits~\cite{knillFermionicLinearOptics2001}, permutation-symmetric systems~\cite{anschuetzEfficientClassicalAlgorithms2023}, and Lie-algebraic simulation when the dynamical Lie algebra is of polynomial dimension~\cite{gohLieAlgebraicClassical2025, barligea2026enablingliealgebraicclassicalsimulation}. These methods are powerful for the structured circuits they target, but the non-Clifford gates and feed-forward measurements central to fault-tolerant magic-state protocols generally fall outside the symmetry classes they exploit. In practice, implementing these methods exposes practical software engineering tradeoffs across total non-Clifford count, qubit count, dynamic tracking overhead, and shot throughput.

    This challenge is especially clear in magic-state preparation protocols. Magic State Cultivation (MSC)~\cite{gidneyMagicStateCultivation2024} has emerged as a particularly compelling target because it promises low-overhead preparation of high-fidelity logical $|T\rangle$ states, and has already motivated both simulation and experimental follow-on work~\cite{rosenfeldMagicStateCultivation2025, chenEfficientMagicState2026, dasuBreakingEvenMagic2025a, sahayFoldtransversalSurfaceCode2026a, hiranoEfficientMagicState2025, vakninEfficientMagicState2026}. At the same time, MSC is a demanding standalone benchmark for simulation. It is representative of an important near-Clifford regime common for quantum error correction (QEC): circuits whose dominant propagation, measurement, and noise structure is Clifford, but whose correctness and performance depend on a relatively small set of non-Clifford operations. Answering the key fault-tolerance questions in this regime requires not only exact treatment of the true non-Clifford circuit, but also sampling at the scale of millions to trillions of shots to evaluate performance at increasingly lower physical noise rates. In the original MSC study, large-scale evaluation  relied on an $S$-gate proxy circuit amenable to simulation in Stim, because an exact simulator for the full noisy non-Clifford circuit was not yet available at useful scale~\cite{gidneyMagicStateCultivation2024}.

    Recent work has made major progress on this frontier. Li~et~al. developed SOFT~\cite{liSOFTHighperformanceSimulator2025} which implemented a generalized stabilizer representation on GPUs to obtain the first ground-truth large-scale simulations of the cultivation circuit at code distance $d=5$. Tuloup and Ayral's Pauli Frame Sparse Representation (PFSR)~\cite{tuloupComputingLogicalError2026a} used a similar sparse non-Clifford structure along with stabilizer-frame tracking, similarly finding that the true $T$-gate cultivation circuit can differ substantially from the $S$-proxy approximation. Tsim~\cite{haenelTsimFastUniversal2026} brings a compile-once, sample-many approach to universal noisy QEC circuits using ZX-calculus reductions and stabilizer-rank methods \cite{wanCuttingStabiliserDecompositions2025, wanSimulatingMagicState2026a}. For easier adoption, Tsim is compatible with the Stim API and extends the Stim circuit format with non-Clifford instructions. Together, these tools demonstrate that the gap beyond pure Clifford simulation is beginning to close. However, these methods still face practical limits as system sizes and shot count increase. For MSC in particular, prior studies omit the final escape stage, leaving a critical gap in end-to-end analysis.

    A complementary contemporaneous direction is SyQMA~\cite{umbrarescuSyQMAMemoryefficientSymbolic2026}, an exact symbolic simulator for universal noisy QEC circuits that augments stabilizer tableaus with auxiliary qubits and a modified trace to compute closed-form expectation values, measurement probabilities, and logical-error-rate expressions. SyQMA targets analytic small-to-moderate QEC-gadget studies such as circuit-level maximum-likelihood decoding and fault-distance verification, whereas the present work focuses on high-throughput exact sampling at the scale needed for end-to-end MSC.

    In this work, we introduce \emph{Clifft}, an open-source simulator designed for this regime. Clifft\footnote{Clifft is a play on "Clifford+T" and pronounced as one syllable.} combines a Stim-like ahead-of-time compilation philosophy with an execution model that confines dense quantum evolution to the active non-Clifford subspace. Like Tsim, Clifft accepts Stim-compatible circuits, including noise, mid-circuit measurements, detector annotations, and classically controlled operations, and adds support for non-Clifford gates. At a high level, Clifft factors the simulated state into a deterministic offline Clifford frame, a virtual Pauli frame for lightweight shot-dependent updates, and a dense active state vector that tracks only the current non-Clifford degrees of freedom. This separates deterministic coordinate transformations from coherent amplitude evolution, bounding dense runtime array operations by $\mathcal{O}(2^{k})$, where $k$ represents the \textit{active dimension} of the state vector. $k$ is dynamic; it expands when non-Clifford gates are applied and contracts when stabilizer measurements collapse the superposition. For many near-Clifford QEC protocols like MSC, frequent interleaved measurements keep these non-Clifford effects spatially and temporally localized. For example, even though the end-to-end MSC protocol of Gidney et al.~\cite{gidneyMagicStateCultivation2024} uses 463 physical qubits through the escape stage, the peak active dimension never exceeds $k_{\max} = 10$.

    While sparse simulators like SOFT and PFSR also scale exponentially with the active dimension, they face online tracking overhead. Their dynamic $\mathcal{O}(N^2)$ tableau eliminations and canonicalization steps become increasingly costly each shot at the full $N=463$ end-to-end scale, helping explain why prior MSC studies stopped after the 42-qubit cultivation stage. In contrast, Clifft evaluates all $\mathcal{O}(N^2)$ tableau updates entirely ahead of time. Runtime, or sample time, execution is thus reduced to lightweight bitwise frame updates and CPU-friendly sweeps over the active state vector array. This shifts the dominant expected per-shot cost to active-array operations scaling as $\mathcal{O}(2^{k_{\max}})$, while confining the remaining dependence on the total qubit count $N$ to lightweight frame, error, and record updates (detailed in Section~\ref{sec:complexity}). In practice, this allows Clifft to sustain hundreds of thousands of shots per second on standard commodity CPUs even for complex fault-tolerant gadgets.

    Our main contributions and results are as follows:
    \begin{itemize}
        \item We introduce a frame-factored state representation that separates offline Clifford compilation, online Pauli-frame tracking, and online dense active-state evolution, yielding exact simulation costs that scale exponentially in the active virtual dimension rather than the full system size.
        \item We develop and open-source Clifft, a Python package with a high-performance C++ core that extends the Stim API to support non-Clifford gates. Translating this frame-factored representation into practice, its custom compiler-to-VM pipeline executes the offline Clifford pass ahead of time, leaving only lightweight Pauli tracking and localized non-Clifford operations for online execution.
        \item We demonstrate that Clifft is competitive across pure Clifford, near-Clifford, and non-Clifford regimes, remaining within an order of magnitude of Stim in the pure-Clifford limit and close to leading CPU state vector simulations in the dense limit, while providing its clearest advantage on low-magic fault-tolerant circuits.
        \item We perform, to our knowledge, the first exact end-to-end simulation of Magic State Cultivation, including the escape stage, over hundreds of billions of shots. These simulations show that escape-stage decoding failures mask the true $T$-gate versus $S$-proxy discrepancy at low decoder-gap thresholds, while at higher thresholds the end-to-end behavior approaches the larger discrepancy observed in the inject and cultivation stages alone.
    \end{itemize}

    Clifft is Apache 2.0 licensed and is available as \texttt{clifft} on PyPI. Clifft also includes a WASM-based interactive playground for exploring circuits, representations and simulations within a standard web browser. Where appropriate, circuits discussed in this paper are linked directly to corresponding playground examples (as \clifftplaydefault\ \ hyperlinks).

   The remainder of this paper is organized as follows. Section~\ref{sec:theory} defines the frame-factored representation and its complexity bounds. Section~\ref{sec:architecture} describes the Clifft compiler, runtime, and validation strategy. Section~\ref{sec:results} benchmarks Clifft across simulation regimes and applies it to exact end-to-end MSC simulation. Section~\ref{sec:conclusion} concludes with a summary of results and future directions.

    \section{Frame-Factored State Representation}
    \label{sec:theory}

    Standard state vector simulation scales exponentially with the total number of qubits in the system. However, fault-tolerant quantum circuits are typically dominated by Clifford operations, which can be tracked efficiently via the Gottesman-Knill theorem \cite{gottesman1998heisenbergrepresentationquantumcomputers}. We utilize this structure by shifting from the standard Schrödinger picture to a hybrid representation that decouples the state into three distinct objects: an offline Clifford coordinate frame, a dynamically sized state vector, and a lightweight Pauli tracking frame.

    By mapping operations into a Pauli tracking frame (which we term the \emph{virtual basis}) and compressing non-Clifford entanglement through a process we call \emph{Pauli localization}, the dense state vector is restricted to a small \emph{active subspace}. Decoupling this discrete coordinate tracking from the continuous amplitude evolution allows exact simulation to scale with the active dimension, $k$, rather than the physical qubit count, $N$. We define this frame-factored representation and its underlying localized virtual geometry in Section~\ref{sec:theory}.

    \begin{figure*}[t]
    \centering
    \begin{tikzpicture}[
        >=stealth, font=\sffamily,
        wire/.style={thick, draw=black!60},
        activewire/.style={thick, draw=red!80!black},
        dormantwire/.style={thick, dashed, draw=gray!60},
        clifford/.style={rectangle, draw=blue!80!black, fill=blue!5, thick, minimum size=0.6cm, rounded corners=0.5ex},
        nonclifford/.style={rectangle, draw=red!80!black, fill=red!5, thick, minimum size=0.6cm, rounded corners=0.5ex},
        vclifford/.style={rectangle, draw=orange!80!black, fill=orange!5, thick, minimum size=0.5cm, rounded corners=0.5ex},
        sectionlabel/.style={font=\bfseries\small, anchor=south, align=center},
        indicator/.style={draw=gray!80, dashed, thick, rounded corners}
    ]

    \node[sectionlabel] at (1.0, 0.8) {Input Circuit};
    \draw[wire] (-0.2, 0) node[left, font=\scriptsize] {$q_0$} -- (2.2, 0);
    \draw[wire] (-0.2, -1) node[left, font=\scriptsize] {$q_1$} -- (2.2, -1);

    \node[clifford, font=\scriptsize] at (0.4, 0) {$H$};

    \filldraw[blue!80!black] (1.0, 0) circle (2pt);
    \draw[thick, blue!80!black] (1.0, 0) -- (1.0, -1);
    \node[circle, draw=blue!80!black, thick, fill=blue!5, inner sep=0pt, minimum size=6pt] (cnot1) at (1.0, -1) {};
    \draw[thick, blue!80!black] (cnot1.north) -- (cnot1.south);
    \draw[thick, blue!80!black] (cnot1.east) -- (cnot1.west);

    \node[nonclifford, font=\scriptsize] (Tgate) at (1.8, -1) {$e^{-i\frac{\pi}{8} Z_1}$};
    \node[font=\scriptsize, above=2pt] at (Tgate.north) {($T$)};

    \draw[->, thick] (2.4, -0.5) -- (3.0, -0.5) node[midway, above, font=\scriptsize, align=center] {Map};

    \node[sectionlabel] at (4.7, 0.8) {Heisenberg Mapped};

    \draw[wire] (3.2, 0) -- (6.3, 0);
    \draw[wire] (3.2, -1) -- (6.3, -1);

    \node[nonclifford, font=\scriptsize, minimum height=1.2cm] (PO1) at (4.3, -0.5) {$e^{-i\frac{\pi}{8} \tilde{X}_0 \otimes \tilde{Z}_1}$};

    \node[clifford, font=\scriptsize, minimum height=1.8cm, minimum width=0.7cm] (UC) at (5.8, -0.5) {$U_C^{(t)}$};

    \draw[->, thick] (6.5, -0.5) -- (7.3, -0.5)
        node[midway, above, font=\scriptsize, align=center] {Localize}
        node[midway, below=2pt, font=\scriptsize, text=orange!80!black, align=center] {$V$};

    \node[sectionlabel] at (8.9, 0.8) {Pauli Localized};

    \draw[indicator] (7.8, 0.6) rectangle (10.0, -1.6);

    \draw[->, dashed, thick, gray!80!black, rounded corners=4pt] (PO1.south) -- (4.3, -1.9) -| (8.9, -1.6);

    \draw[wire] (7.5, 0) -- (7.8, 0);

    \draw[activewire] (7.8, 0) -- (10.5, 0) node[right, font=\scriptsize, text=red!80!black] {$A$};

    \draw[wire] (7.5, -1) -- (7.8, -1);
    \draw[dormantwire] (7.8, -1) -- (10.5, -1) node[right, font=\scriptsize, text=gray!80!black] {$D$};

    \node[nonclifford, font=\scriptsize] at (8.9, 0) {$e^{-i\frac{\pi}{8} \tilde{Z}_0}$};
    \node[font=\tiny, text=gray!80!black] at (8.9, -0.8) {$k \to k+1$};

    \draw[->, thick] (11.0, -0.5) -- (11.7, -0.5) node[midway, above, font=\scriptsize, align=center] {Measure};

    \node[sectionlabel] at (13.2, 0.8) {Projected State};

    \draw[indicator] (12.0, 0.6) rectangle (14.3, -1.6);

    \draw[activewire] (11.9, 0) -- (13.0, 0);

    \draw[dormantwire] (13.0, 0) -- (14.7, 0);
    \draw[dormantwire] (11.9, -1) -- (14.7, -1);

    \node[nonclifford, font=\scriptsize] at (13.0, 0) {$\Pi_m$};

    \node[font=\tiny, text=gray!80!black] at (13.15, -0.8) {$k \to k-1$};

    \end{tikzpicture}
    \caption{Lifecycle of an active operation in the frame-factored representation. Physical Clifford gates are absorbed into the offline Clifford frame $U_C$, while a non-Clifford operation is Heisenberg mapped into a virtual Pauli rotation. Pauli localization uses virtual Clifford transformations to reduce the mapped multi-qubit Pauli string to a single virtual axis. The resulting localized operation either updates the active state vector or promotes a dormant virtual qubit into the active set, increasing $k$. Later projective measurements can collapse an active axis and return it to the dormant set, decreasing $k$. Stochastic Pauli noise and feed-forward corrections update only the virtual Pauli frame $\tilde{P}$ (not shown).}
    \label{fig:localization_pipeline}
\end{figure*}
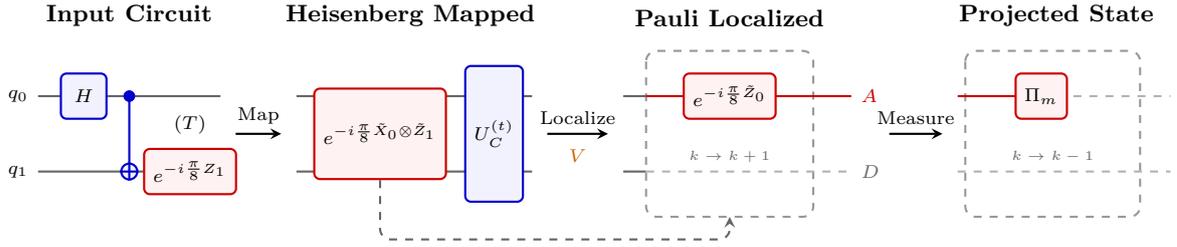

    \subsection{Circuit Model and the Heisenberg Picture}
    \label{sec:heisenberg}

    We consider a quantum circuit over $N$ qubits as a sequence of operations applied at discrete timesteps $t = 0,1,\ldots$. Let $P \in \mathcal{P}_N$ be an $N$-qubit Pauli string operator, where $\mathcal{P}_N$ denotes the $N$-qubit Pauli group. We partition circuit operations into two classes:
    \begin{enumerate}
        \item \textbf{Passive operations:} Clifford unitaries $C \in \mathcal{C}_N$, where the Clifford group $\mathcal{C}_N = \{U : U \mathcal{P}_N U^\dagger = \mathcal{P}_N\}$ is the normalizer of the Pauli group.
        \item \textbf{Active operations:} Non-Clifford Pauli rotations, projective Pauli measurements, stochastic Pauli noise, and conditionally applied Pauli corrections.
    \end{enumerate}
    This model supports standard Clifford+$T$ circuits, Pauli noise, mid-circuit measurements, and classical feed-forward. More general non-Clifford unitary gates can be decomposed into Pauli rotations before evaluation.

    To isolate the non-Clifford complexity, we use a hybrid Heisenberg-picture formulation. Since Clifford gates normalize the Pauli group, a physical Clifford gate applied to the state can instead be treated as a passive update to the coordinate system. Let $U_C^{(t)}$ be the cumulative Clifford frame representing these coordinate transformations up to time $t$. For any active operation $O$ applied in the physical basis at time $t$, its Pauli generator $P_O$ is mapped into the virtual basis by conjugating through this frame:
    \begin{definition}[Heisenberg Mapping]
    \label{def:heisenberg_map}
    \begin{equation}
        \label{eq:heisenberg_map}
        \tilde{P}_O = (U_C^{(t)})^\dagger P_O U_C^{(t)}.
    \end{equation}
    Here $\tilde{P}_O$ is the virtual Pauli generator corresponding to the physical operation $O$ in the Clifford frame at time $t$.
    \end{definition}

    Because active operations in the supported circuit model are generated by local Pauli operators, the mapped generator $\tilde{P}_O$ remains a valid Pauli string, although it may act on many virtual qubits. This transformation is conceptually similar to the circuit transformations used in Pauli-based computation approaches~\cite{bravyiTradingClassicalQuantum2016,litinskiGameSurfaceCodes2019a} and Pauli-centric compiler optimization frameworks~\cite{liPaulihedralGeneralizedBlockWise2021, paykinPCOASTPaulibasedQuantum2023, wangTableauBasedFramework2025}.

    \subsection{The Frame-Factored State}
    \label{sec:representation}

    The Heisenberg mapping lets all active operations be expressed in the virtual basis induced by the Clifford frame. We represent the physical quantum state at time $t$ as:
    \begin{definition}[Frame-Factored State]
    \label{def:frame_decomposed}
    \begin{equation}
        \label{eq:factored_state}
        |\psi^{(t)}\rangle =
        \gamma^{(t)} U_C^{(t)} \tilde{P}^{(t)}
        \Big( |\phi^{(t)}\rangle_A \otimes |0\rangle_D \Big),
    \end{equation}
    where:
    \begin{itemize}
        \item $U_C^{(t)} \in \mathcal{C}_N$ is the \textbf{Clifford frame}, a deterministic unitary mapping from the virtual basis to the physical laboratory basis.
        \item $\tilde{P}^{(t)}$ is the \textbf{virtual Pauli frame}, an $N$-qubit phase-free Pauli operator parameterized by binary vectors $\mathbf{x}^{(t)}, \mathbf{z}^{(t)} \in \{0,1\}^N$:
        \begin{equation}
            \tilde{P}^{(t)} =
            \bigotimes_{j=0}^{N-1}
            X_j^{x_j^{(t)}} Z_j^{z_j^{(t)}} .
        \end{equation}
        \item $A$ and $D$ are disjoint sets of \textbf{active} and \textbf{dormant} virtual qubits, with $|A|=k$ and $|D|=N-k$.
        \item $|\phi^{(t)}\rangle_A \in \mathbb{C}^{2^k}$ is the \textbf{active state vector}, which stores the continuous amplitudes associated with the active virtual subspace.
        \item $\gamma^{(t)} \in \mathbb{C}$ is a \textbf{global scalar} tracking global phase and normalization.
    \end{itemize}
    \end{definition}

    The dormant subspace is fixed to $|0\rangle_D$ in the virtual basis. All coherent superposition not captured by the Clifford and Pauli frames is therefore confined to the active state vector. The maximum active dimension over execution,
    \begin{equation}
        k_{\max} = \max_t |A^{(t)}|,
    \end{equation}
    bounds the dominant exponential cost of simulation. For near-Clifford protocols with frequent measurements, $k_{\max}$ can be much smaller than the total number of physical qubits $N$.

    At initialization, $U_C^{(0)}=I$, $\tilde{P}^{(0)}=I$, $A=\emptyset$, $D=\{0,\ldots,N-1\}$, $|\phi^{(0)}\rangle_A=[1]$, and $\gamma^{(0)}=1$.

    \subsection{Virtual Operations and Pauli Localization}
    \label{sec:virtual_ops}

    Because evaluating the amplitudes of the active state vector scales exponentially with the dimension $k$, the primary objective is to localize non-Clifford evolution to as few active virtual qubits as possible. We achieve this via \emph{Pauli localization}.

    Suppose an active operation transforms to a multi-qubit virtual Pauli generator $\tilde{P}_O$ under the Heisenberg mapping described in Equation~\ref{eq:heisenberg_map}. We construct a sequence of virtual Clifford gates $V \in \mathcal{C}_N$ to map this multi-qubit operator to a single-qubit one:

    \begin{lemma}[Pauli Localization]
        \label{lem:pauli_localization}
        For any non-identity $N$-qubit virtual Pauli operator $\tilde{P}_O$, there exists a sequence of at most $2N$ virtual Clifford gates $V \in \mathcal{C}_N$ such that $V \tilde{P}_O V^\dagger = \alpha P_v$, where $P_v \in \{X_v, Y_v, Z_v\}$ acts non-trivially on exactly one virtual qubit $v$, and $\alpha \in \{\pm 1, \pm i\}$.
    \end{lemma}

    A constructive proof detailing the $\mathcal{O}(N)$ greedy localization algorithm is provided in Appendix~\ref{app:pauli_localization}.

    To integrate this localization into the frame factorization without altering the physical state, we explicitly track how $V$ distributes. First we observe the the result of Heisenberg mapping a physical operation $\exp(-i\theta P_O)$ on the frame-factored state~\ref{eq:factored_state} yields (ignoring the $(t)$ superscript):
    \begin{align}
        \exp(-i\theta P_O) |\psi\rangle &= \exp(-i\theta P_O) \gamma U_C \tilde{P} \Big( |\phi\rangle_A \otimes |0\rangle_D \Big) \nonumber \\
        &= \gamma U_C \Big( U_C^\dagger \exp(-i\theta P_O) U_C \Big) \tilde{P} \Big( |\phi\rangle_A \otimes |0\rangle_D \Big) \nonumber \\
        &= \gamma U_C \exp(-i\theta \tilde{P}_O) \tilde{P} \Big( |\phi\rangle_A \otimes |0\rangle_D \Big)
    \end{align}
    Following Pauli localization, inserting the identity $V^\dagger V = I$ between the operators factors the expression further:
    \begin{align}
        \exp(-i\theta P_O) |\psi\rangle &= \gamma U_C V^\dagger \Big( V \exp(-i\theta \tilde{P}_O) V^\dagger \Big) \Big( V \tilde{P} V^\dagger \Big) V \Big( |\phi\rangle_A \otimes |0\rangle_D \Big) \nonumber \\
        &= \gamma \Big( U_C V^\dagger \Big) \exp(-i\theta \alpha P_v) \Big( V \tilde{P} V^\dagger \Big) \Big[ V \Big( |\phi\rangle_A \otimes |0\rangle_D \Big) \Big]
    \end{align}

    This effectively absorbs the inverse transformation into the offline Clifford coordinate frame ($U_C \leftarrow U_C V^\dagger$), conjugates the Pauli frame ($\tilde{P} \leftarrow V \tilde{P} V^\dagger$), and isolates the operation $\exp(-i\theta \alpha P_v)$ to a single axis of the state vector.

    The useful case is when $V$ can be chosen so that it does not create superposition in dormant qubits. Clifft therefore prefers localization sequences that act trivially on $|0\rangle_D$ whenever possible.
    \begin{lemma}[Dormant Invariance]
        \label{lem:dormant_invariance}
        Let $V \in \mathcal{C}_N$ be a sequence of controlled-Pauli operations whose controls are restricted to dormant qubits in $D$. Then $V$ acts as the identity on the computational-zero dormant subspace:
        \begin{equation}
            V\Big(|\phi\rangle_A \otimes |0\rangle_D\Big)
            =
            |\phi\rangle_A \otimes |0\rangle_D .
        \end{equation}
    \end{lemma}

    When this condition holds, localization changes only the coordinate and Pauli frames, not the active state vector. When it does not, the localized operation may promote a dormant virtual qubit into the active set. The next subsection summarizes how the supported active operations update the frame-factored state.

    \subsection{Active Operation Updates}
    \label{sec:active_updates}

    Once an operation has been mapped into the virtual basis and localized to a single virtual axis, its effect on the frame-factored state falls into three categories: continuous rotations, projective measurements, and conditional Pauli operations.

    \textbf{Continuous Pauli rotations.}
    A non-Clifford Pauli rotation $\exp(-i\theta P_O)$ is mapped and localized to a single-axis virtual rotation $\exp(-i\theta' P_v)$. Before acting on the active state vector, the localized rotation is commuted past the virtual Pauli frame. Since Pauli operators either commute or anticommute, this only changes the sign of the rotation angle according to a boolean parity of $\tilde{P}$.

    The resulting update depends on whether $v$ is dormant or active. If $v \in D$ and the localized rotation is diagonal on the dormant $|0\rangle_v$ state, the rotation contributes only a scalar phase to $\gamma$. If evaluating the rotation requires placing $v$ into a conjugate basis, i.e. $|+\rangle_v$, then $v$ is promoted into the active set, $A \leftarrow A \cup \{v\}$, and the active state vector expands from dimension $2^k$ to $2^{k+1}$. If $v \in A$ already, the rotation is applied directly to the corresponding tensor factor of $|\phi\rangle_A$ without changing $k$. The detailed case analysis is given in Appendix~\ref{sec:continuous_rotations}.

    \textbf{Projective Pauli measurements.}
    A projective Pauli measurement is similarly mapped and localized to a single virtual observable $M_v \in \{X_v,Z_v\}$ with projector
    \begin{equation}
        \Pi_m = \frac{1}{2}\bigl(I + (-1)^m M_v\bigr),
    \end{equation}
    where $m \in \{0,1\}$ is the physical measurement outcome. However, we do not need to localize this mapped projector. Instead, commuting it past the virtual Pauli frame extracts a parity shift:
    \begin{equation}
        \Pi_m \tilde{P}
        =
        \tilde{P}\Pi_{m \oplus p},
    \end{equation}
    where $p=1$ if $M_v$ anticommutes with $\tilde{P}$ and $p=0$ otherwise.

    If $v \in D$, this final projective measurement is resolved against a known virtual $|0\rangle_v$ state. Depending on the measurement basis, the outcome is either deterministic up to the Pauli-frame shift or uniformly random with a deterministic basis update absorbed into $U_C$ and $\tilde{P}$. In either case the active state vector is unchanged. If $v \in A$, the shifted projector is applied directly to $|\phi\rangle_A$. The measured axis then becomes disentangled from the remaining active state, can be returned to the dormant set, and the active dimension contracts as $k \leftarrow k-1$. A more explicit case-by-case description appears in Appendix~\ref{app:measurement_updates}.

    \textbf{Conditional Pauli operations.}
    Stochastic Pauli noise, feed-forward corrections, and other classically controlled Pauli operations are also handled without Pauli localization. After Heisenberg mapping, a conditional physical Pauli $E^c$ becomes a conditional virtual Pauli $\tilde{E}^c$, where the condition $c \in \{0,1\}$ is determined by either a random sample or a previous measurement result. Because the Pauli group is closed under multiplication, the update is absorbed into the virtual Pauli frame:
    \begin{equation}
        \tilde{P} \leftarrow \tilde{E}^c \tilde{P},
    \end{equation}
    with any phase accumulated into $\gamma$. These operations do not change $U_C$, do not change $A$, and do not traverse the active state vector. Details are given in Appendix~\ref{sec:conditional_paulis}.

    \subsection{Coordinate-Amplitude Decoupling}
    \label{sec:decoupling_theorem}

    The frame-factored representation separates deterministic coordinate evolution from sample-dependent amplitude updates. This separation is the property that allows Clifft to compile the Clifford geometry once and reuse it across many shots.

    \begin{theorem}
    \label{thm:decoupling}
    Under the frame-factored representation, the trajectory of the Clifford frame $U_C^{(t)}$ and the active-set geometry $A^{(t)}$ are determined by the circuit structure and localization choices. They are independent of stochastic Pauli error samples, probabilistic measurement outcomes, and the complex amplitudes of the active state vector $|\phi^{(t)}\rangle_A$.
    \end{theorem}

    \begin{proof}
    By Definition~\ref{def:frame_decomposed}, all physical Clifford gates are absorbed into the Clifford frame $U_C^{(t)}$. By Definition~\ref{def:heisenberg_map}, each active operation has a virtual Pauli generator determined algebraically by the current Clifford frame and the physical circuit operation. The localization sequence $V$ in Lemma~\ref{lem:pauli_localization} is chosen from this virtual generator and the current active/dormant layout, not from the numerical amplitudes of $|\phi\rangle_A$.

    Continuous Pauli rotations can expand the active set only when their localized virtual axis requires promotion from the dormant subspace; this condition is fixed by the mapped generator and localization choice. Projective measurements can contract the active set only when their localized virtual axis lies in $A$; the measured eigenvalue affects the branch and normalization but not which axis is measured or demoted. Conditional Pauli operations, including stochastic noise and feed-forward corrections, multiply into the virtual Pauli frame and leave $U_C$ and $A$ unchanged. Therefore the evolution of $U_C^{(t)}$ and $A^{(t)}$, and hence the bound $k_{\max}$, can be determined before sampling.
    \end{proof}

    This structural guarantee changes the execution model for exact simulation. All Clifford coordinate updates, Heisenberg mappings, and active-space planning can be resolved ahead of time. Runtime sampling then performs only Pauli-frame updates, active-array operations, and measurement/noise sampling over a fixed compiled schedule.

    \subsection{Asymptotic Complexity}
    \label{sec:complexity}

    The frame-factored representation shifts the dominant exponential cost of simulation from the total number of physical qubits $N$ to the peak active virtual dimension $k_{\max}$. Let $C$ denote the number of Clifford operations, $M$ the number of measurements, $T$ the number of non-Clifford rotations, and $E$ the number of stochastic Pauli error mechanisms.

    \textbf{Compile-time cost.}
    All deterministic Clifford coordinate transformations and Pauli localization steps are resolved ahead of time. Absorbing the $C$ physical Clifford operations into the coordinate frame via tableau updates scales as $\mathcal{O}(CN)$\footnote{The implementation uses packed bit operations for many tableau, Pauli-frame, and record updates, so some linear factors in $N$ are reduced by machine-word parallelism in practice. Throughout, we write these costs as $\mathcal{O}(N)$ to emphasize the algorithmic dependence on the full system size rather than word-level constants.}. Mapping the $E$ error mechanisms into the virtual basis requires $\mathcal{O}(EN)$ work. Measurements and non-Clifford rotations also undergo Pauli localization; synthesizing and absorbing the resulting virtual Clifford transformations gives a worst-case cost of $\mathcal{O}((M+T)N^2)$. The total offline cost is therefore bounded by
    \begin{equation}
        \mathcal{O}\bigl(CN + EN + (M+T)N^2\bigr),
    \end{equation}
    and is incurred once per circuit.

    \textbf{Sample-time cost.}
    By Theorem~\ref{thm:decoupling}, the Clifford frame and active-set schedule are fixed before sampling. Per-shot execution is therefore independent of the number of passive Clifford gates $C$ and consists of two components.

    Discrete tracking and sampling includes Pauli-frame updates, classical control, detector and post-selection logic, and stochastic error sampling. These operations act on bit representations over $N$ qubits and therefore contribute
    \begin{equation}
        \mathcal{O}\bigl((T + M + E)N\bigr)
    \end{equation}
    in the worst case. In practice, this bookkeeping is implemented using packed bit operations and is not the dominant runtime cost in the near-Clifford regimes studied here. Sparse or gap-based noise sampling can further reduce the effective dependence on $E$ when faults are rare.

    Dense active-subspace evolution includes non-Clifford rotations and active measurements that traverse the active state vector. This contributes
    \begin{equation}
        \mathcal{O}\bigl((T + M_{\mathrm{active}})2^{k_{\max}}\bigr),
    \end{equation}
    where $M_{\mathrm{active}}$ is the number of measurements that act on the active subspace.

    The total worst-case per-shot cost is therefore
    \begin{equation}
        \label{eq:per_shot_cost}
        \mathcal{O}\bigl((T + M + E)N + (T + M_{\mathrm{active}})2^{k_{\max}}\bigr).
    \end{equation}
    In the target near-Clifford regimes, runtime is typically dominated by active-array traversals of size $2^{k_{\max}}$, while the remaining dependence on the full physical qubit count $N$ is confined to lightweight frame, error, and record updates.

    \section{Clifft Compilation and Execution}
    \label{sec:architecture}

    Section~\ref{sec:theory} showed that deterministic Clifford-coordinate evolution can be separated from runtime amplitude updates. Clifft turns this separation into a compile-once, sample-many execution model. The compiler resolves Clifford rewrites, virtual-coordinate tracking, Pauli localization, and active-space planning ahead of time. Each shot then executes a fixed schedule consisting only of Pauli-frame updates, stochastic sampling, classical control, and localized operations on the active state vector.

    Figure~\ref{fig:clifft_architecture} summarizes the pipeline. The input circuit is first lowered into the Heisenberg Intermediate Representation (HIR), which absorbs Clifford gates into the offline frame and maps the generators of active operations into the virtual basis. The compiler then optimizes this representation, applies Pauli localization, and emits bytecode for the Schr\"odinger Virtual Machine (SVM). The SVM executes this bytecode over the factored runtime state to generate samples efficiently.

    The remainder of this section describes the compiler pipeline, the SVM execution model, and the validation strategy used to test the implementation. For a full listing of supported instructions, bytecode operations, and compiler optimizations, see the documentation and codebase \cite{clifftDocsCode2026}.

    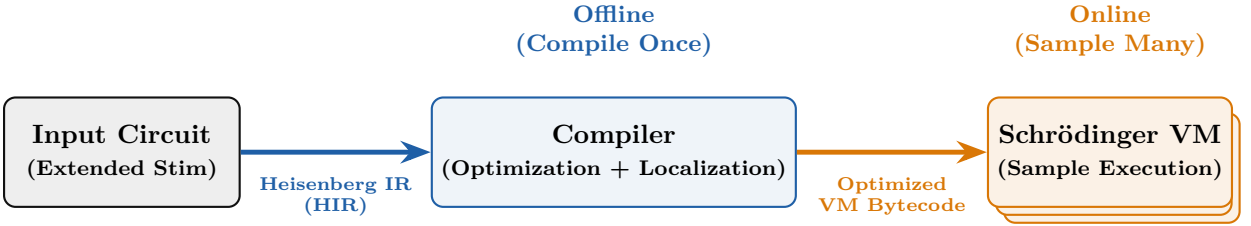
\begin{figure*}[t]
    \centering
    \begin{tikzpicture}[
            >=Stealth,
            font=\sffamily,
            node distance=2.5cm,
            simple-io/.style={rectangle, draw=UFBlack, fill=UFGray!10, thick, minimum width=3.1cm, minimum height=1.45cm, align=center, rounded corners=1ex, font=\normalsize\bfseries},
            simple-compiler/.style={rectangle, draw=UFBlue, fill=UFBlue!7, thick, minimum width=3.1cm, minimum height=1.45cm, align=center, rounded corners=1ex, font=\normalsize\bfseries},
            simple-vm/.style={rectangle, draw=UFOrange, fill=UFOrange!10, thick, minimum width=3.1cm, minimum height=1.45cm, align=center, rounded corners=1ex, font=\normalsize\bfseries},
            arr-hir/.style={->, line width=0.65mm, draw=UFBlue},
            arr-vm/.style={->, line width=0.65mm, draw=UFOrange},
            ann/.style={font=\scriptsize\bfseries, align=center},
            phase-label/.style={font=\small\bfseries, align=center}
        ]

        \node[simple-io] (input) {Input Circuit\\ {\fontsize{8}{9}\selectfont\bfseries (Extended Stim)}};
        \node[simple-compiler, right=of input] (compiler) {Compiler \\ {\fontsize{8}{9}\selectfont\bfseries (Optimization + Localization)}};

        \node[simple-vm, right=of compiler, xshift=6pt, yshift=-6pt] (svm_bg2) {};
        \node[simple-vm, right=of compiler, xshift=3pt, yshift=-3pt] (svm_bg1) {};
        \node[simple-vm, right=of compiler] (svm) {Schr\"odinger VM \\ {\fontsize{8}{9}\selectfont\bfseries (Sample Execution)}};

        \node[phase-label, text=UFBlue, above=0.4cm of compiler] {Offline\\(Compile Once)};
        \node[phase-label, text=UFOrange, above=0.4cm of svm] {Online\\(Sample Many)};

        \draw[arr-hir] (input.east) -- node[below=5pt, ann, text=UFBlue] {Heisenberg IR\\ (HIR)} (compiler.west);
        \draw[arr-vm] (compiler.east) -- node[below=5pt, ann, text=UFOrange] {Optimized\\VM Bytecode} (svm.west);

    \end{tikzpicture}
    \caption{The high-level Clifft execution pipeline. The compiler resolves Clifford-coordinate evolution, HIR optimization, Pauli localization, and active-space planning once. The SVM then samples the compiled bytecode many times over the frame-factored runtime state.}
    \label{fig:clifft_architecture}
    \end{figure*}

    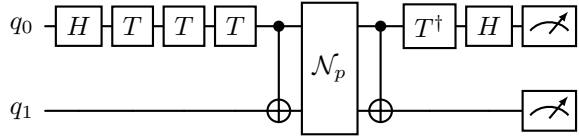
\begin{figure}[t]
    \centering
    \begin{quantikz}[column sep=4pt]
    \lstick{$q_0$} & \gate{H} & \gate{T} & \gate{T} & \gate{T} & \ctrl{1} & \gate[2]{\mathcal{N}_p} & \ctrl{1} & \gate{T^\dagger} & \gate{H} & \meter{} \\
    \lstick{$q_1$} &          &          &          &          & \targ{}  &                          & \targ{}  &                  &          & \meter{}
    \end{quantikz}
    \caption{Example mirror circuit used to illustrate Clifft's multi-level lowering. The unitary block $\mathrm{CX}\cdot T_0^3\cdot H_0$ and its simplified inverse sandwich a depolarizing channel $\mathcal{N}_p$ at strength $p$. The initial gates are intentionally written in an unsimplified form to demonstrate compiler optimization. The same example can be explored interactively in the \href{https://unitaryfoundation.github.io/clifft/playground/?code=MQAgsglgTlD2UgMbUQVwgFwFAAkQAYsAVA40kwgYQA0CQBGLAEQFEAFAeQBkBBAJQCSALRb0AFPgB0+fPQCUdRjUXEA+kx4BxUnkJhFQA}{\faPlay\ Clifft playground}.}
    \label{fig:mirror_demo}
    \end{figure}

    \subsection{Compilation Pipeline}
    \label{sec:compilation_pipeline}

    The ahead-of-time compiler translates a physical circuit into optimized SVM bytecode in four stages. The first two stages construct and optimize the Heisenberg Intermediate Representation (HIR), while the last two specialize that representation to the frame-factored runtime state by applying Pauli localization and bytecode optimization.

    We illustrate these stages using the mirror circuit in Figure~\ref{fig:mirror_demo}. The lowering shown in Figure~\ref{fig:pipeline_demo} displays the optimized artifact at each level. The same example can be inspected interactively in the playground linked from Figure~\ref{fig:mirror_demo}, including the unoptimized forms omitted here.

    \begin{figure*}[t]
    \centering
    \small
    \begin{minipage}[t]{0.29\textwidth}
    \centering\textbf{Stim-extended input}
    \begin{verbatim}
    H 0
    T 0
    T 0
    T 0
    CX 0 1
    DEPOLARIZE1(0.001) 0 1
    CX 0 1
    T_DAG 0
    H 0
    M 0 1
    \end{verbatim}
    \end{minipage}\hfill
    $\xrightarrow{\text{\scriptsize HIR}}$\hfill
    \begin{minipage}[t]{0.25\textwidth}
    \centering\textbf{Optimized HIR}
    \begin{verbatim}
    T     +X0
    NOISE site=0
    NOISE site=1
    MEAS  +Z1 -> rec[1]
    T_DAG +X0
    MEAS  +Y0 -> rec[0]
    \end{verbatim}
    \end{minipage}\hfill
    $\xrightarrow{\text{\scriptsize SVM}}$\hfill
    \begin{minipage}[t]{0.31\textwidth}
    \centering\textbf{Optimized bytecode}
    \begin{verbatim}
    OP_FRAME_H 0
    OP_EXPAND_T 0
    OP_NOISE_BLOCK sites=[0..2)
    OP_MEAS_DORMANT_STATIC 1
        -> rec[1]
    OP_ARRAY_T_DAG 0
    OP_ARRAY_S 0
    OP_MEAS_ACTIVE_INTERFERE 0
        -> rec[0]
    \end{verbatim}
    \end{minipage}
    \caption{A compact view of Clifft's multi-level lowering for the example in Figure~\ref{fig:mirror_demo}. HIR exposes active operations in the virtual basis after Clifford absorption and algebraic optimization. SVM bytecode then specializes the optimized HIR to the frame-factored runtime state, exposing frame updates, active-state expansion, active-array operations, dormant and active measurements, and coalesced noise sampling.}
    \label{fig:pipeline_demo}
    \end{figure*}

    \textbf{1. Front-End: Heisenberg mapping.}
    The front-end accepts circuits in the Stim format, extended with non-Clifford operations similarly to Tsim, including \texttt{T}, arbitrary-angle Pauli rotations such as \texttt{R\_X}, and single-qubit rotations such as \texttt{U3}\cite{haenelTsimFastUniversal2026}. Nearly all of Stim's existing noise channels, mid-circuit measurements, detectors, observables, and repeat blocks are supported without modification. Clifft integrates Stim's optimized C++ tableau implementation directly and leverages its inverse tableau tracking to absorb physical Clifford gates into the coordinate frame $U_C$ while mapping active operations into the virtual basis via the Heisenberg mapping of Definition~\ref{def:heisenberg_map}.

    The result is the Heisenberg Intermediate Representation (HIR), an equivalent representation of the original circuit expressed in the virtual basis. HIR nodes correspond to the active circuit operations described in Section~\ref{sec:heisenberg}: non-Clifford Pauli rotations, projective Pauli measurements, stochastic Pauli noise, and classically controlled Pauli operations. Passive Clifford gates do not appear explicitly in HIR; their effect survives through the transformed virtual generators $\tilde{P}_O$ of downstream active operations.

    \textbf{2. Middle-End: HIR optimization.}
    Once active operations are expressed in a common virtual basis, HIR provides a compact algebraic optimization layer. Each Pauli generator is stored as a binary symplectic bitstring, so commutation tests reduce to bitwise parity checks rather than explicit operator algebra.

    Clifft currently applies two HIR optimization passes. The first is a peephole pass that reorders commuting operations, fuses adjacent phases, cancels inverses, and absorbs any resulting Clifford action back into the offline frame. The second reorders commuting measurements and non-Clifford operations to push non-Cliffords later and pull measurements earlier when possible. This reduces the peak active dimension $k_{\max}$, or more generally the total time spent at larger active dimensions.

    In Figure~\ref{fig:pipeline_demo}, these passes combine three $T$ rotations into a single non-Clifford rotation, absorb the residual Clifford action into the frame, and update downstream virtual generators accordingly. The optimizer also moves the \texttt{MEAS +Z1} operation before the later $T^\dagger$ rotation, illustrating the same scheduling principle used more generally to reduce active-state work.

    These HIR optimizations are not specific to the SVM. They act on the Heisenberg representation of the circuit and could, in principle, target other compilation or simulation back ends.

    \textbf{3. Back-End: Pauli localization.}
    The back-end lowers optimized HIR into a bytecode specialized to the SVM's frame-factored state. As it processes each HIR operation, it applies Pauli localization as described in Lemma~\ref{lem:pauli_localization}. This converts mapped multi-qubit virtual generators into localized operations on individual virtual axes and determines whether those axes are active or dormant at runtime.

    The bytecode distinguishes the resulting frame updates, active-state expansion, active-array operations, dormant measurements, active measurements, noise sampling, and classical record updates. In Figure~\ref{fig:pipeline_demo}, the optimized bytecode contains both \texttt{OP\_MEAS\_DORMANT\_STATIC}, which resolves a measurement against a known dormant virtual qubit, and \texttt{OP\_MEAS\_ACTIVE\_INTERFERE}, which acts on the active state vector and can collapse an active axis. These distinctions are precisely the runtime counterpart of the active/dormant decomposition in Definition~\ref{def:frame_decomposed}.

    \textbf{4. Bytecode optimization.}
    After initial lowering, a final pass manager rewrites SVM bytecode to reduce instruction dispatch overhead and minimize full traversals of the active state vector. These optimizations target the low-level runtime costs identified in Section~\ref{sec:complexity}: active-array operations scale as $2^k$, while noise sampling and classical bookkeeping should avoid work proportional to the full number of inactive events.

    In Figure~\ref{fig:pipeline_demo}, expansion of the active state and the first $T$ update are fused into a single \texttt{OP\_EXPAND\_T} instruction, avoiding a separate active-array traversal. The two adjacent noise instructions are also coalesced into a single \texttt{OP\_NOISE\_BLOCK}, reducing dispatch overhead and allowing the SVM's noise sampler to process the block as one unit. Unlike HIR optimizations, which simplify the quantum algebra of the circuit, bytecode optimizations target the concrete execution costs of the SVM.

    \subsection{Sampling with the Schr\"odinger Virtual Machine}
    \label{sec:svm}

    The Schr\"odinger Virtual Machine executes the compiled bytecode over the frame-factored runtime state. Each program instance allocates an active array of size $2^{k_{\max}}$, bit vectors for the virtual Pauli frame, a complex scalar for global phase and normalization, measurement record buffers, and state for the random number generators and noise sampler. Each shot reuses this storage: the SVM resets the active state and Pauli frame, executes the bytecode, and records the requested outputs.

    This fixed-layout execution model avoids per-shot allocation because the compiler has already determined $k_{\max}$ and the active-set schedule. It also makes sampling embarrassingly parallel across shots or shot batches. The SVM implementation is organized around the main runtime bottlenecks implied by Section~\ref{sec:complexity}.

    \textbf{Active-state vector bottleneck.}
    The dominant cost of non-Clifford execution is the traversal of the active state vector. The SVM therefore treats the active array as its primary optimization target. Single-shot complex arithmetic over the active subspace is implemented using Single-Instruction Multiple-Data (SIMD) kernels, and the bytecode optimizer reduces the number of full-array traversals by fusing compatible operations before execution. This use of SIMD differs from Stim's pure-Clifford strategy: Stim can vectorize Pauli-frame bit operations across many shots, whereas Clifft uses SIMD primarily within the dense active array of a single shot.

    When the active state is small enough to fit in cache, execution remains single-threaded to preserve locality and avoid synchronization overhead. Once $k$ becomes large enough for memory bandwidth and vector throughput to dominate, the SVM uses OpenMP to parallelize active-array kernels. In the current implementation this transition occurs at $k>18$.

    \textbf{Sparse error bottleneck.}
    In the low-noise regime ($p \ll 1$), most stochastic error locations do not trigger. Naively drawing a Bernoulli random variable at every site therefore wastes work on non-events. Similar in spirit to Stim's low-entropy noise-sampling optimization~\cite{gidneyStimFastStabilizer2021}, Clifft samples from a cumulative hazard array and jumps directly to the next realized fault. This supports heterogeneous noise probabilities while making the expected cost of stochastic noise handling scale with the number of realized faults rather than the total number of noise sites.

    \textbf{Rare-event sampling overhead.}
    For low logical error rates, brute-force Monte Carlo requires a number of shots scaling as $\mathcal{O}(1/p_L)$ to resolve a logical error rate $p_L$. This is a statistical bottleneck of rare-event estimation, not a limitation introduced by Clifft's frame-factored architecture. However, because such rare-event estimates are central to fault-tolerant simulation, Clifft's architecture enables native support for the stratified importance-sampling approach used by Tuloup and Ayral~\cite{tuloupComputingLogicalError2026a}. Shot batches are conditioned on a fixed number of realized faults and weighted by the exact probability of that fault count. For heterogeneous error probabilities, Clifft computes the corresponding Poisson-binomial distribution and marks noise sites as forced to trigger or not trigger for each shot. This allows the same SVM bytecode to generate importance-weighted samples with minimal additional runtime overhead.

    \textbf{Post-selection overhead.}
    Classical operations are evaluated inline during bytecode execution. Detector recording, feed-forward Pauli corrections, post-selection checks, and expectation-value probes operate directly on the frame-factored runtime state. For high-discard protocols such as magic state cultivation, the compiler can insert \texttt{OP\_POSTSELECT} instructions that terminate rejected shots immediately. This avoids executing later bytecode, including potentially expensive active-array operations, for shots that no longer contribute to the kept ensemble.

    \subsection{ Validation}
    \label{sec:validation}

    We validate Clifft using complementary strategies rather than relying on random circuit fuzzing alone. Deep random circuits quickly wash structured amplitude information into effectively featureless output statistics, making them a poor probe for subtle phase, interference, and frame-tracking errors. Instead, we combine interface contract tests, structured circuit families with known behavior, and cross-validation against independent simulators.

    At the interface level, Clifft relies on Stim for tableau operations and inverse Pauli rewinding. We therefore rely on Stim's own extensive testing for those primitives, and add contract tests that codify the specific assumptions Clifft makes about Stim's tableau and Heisenberg-rewinding behavior. These tests ensure that future changes in Stim internals do not silently alter the frame dynamics used by the Clifft compiler.

    We then validate the core simulation logic using structured circuits. Mirror circuits of the form $U U^\dagger = I$ are especially useful because they combine Clifford layers, non-Clifford gates, and entangling structure while retaining a simple expected outcome: noiseless simulations must return deterministically to the initial computational-basis state. These tests exercise active-state expansion and contraction, phase accumulation, frame tracking, and numerical reversibility. We also use topological stress tests, including fan-out patterns and commutation-heavy constructions, to check that optimizer rewrites and bytecode fusions preserve circuit behavior.

    Finally, we cross-check Clifft against external simulators in three regimes. First, for sufficiently small circuits, we explicitly expand Clifft's factored state $\ket{\psi}=\gamma U_C \tilde{P}\ket{\phi}_A$ into a dense computational-basis state vector and compare it against the Qiskit Aer state-vector simulator~\cite{qiskit2024}. Second, on Clifford-compatible noisy QEC circuits, we compare detector and logical-observable statistics against Stim over millions of shots, checking that observed marginals agree within binomial shot-noise bounds. Third, we inject deterministic Pauli faults into topological QEC circuits and verify that Clifft produces the same logical detector trajectories as Stim, confirming that the compiled virtual-frame dynamics reproduce the correct physical error propagation.

    \section{Benchmarks and Results}
    \label{sec:results}

    In this section, we evaluate Clifft across a range of circuit regimes to characterize its performance and scaling behavior. We begin with controlled benchmarks spanning pure-Clifford, near-Clifford, and dense non-Clifford circuits. We then turn to a study of the Magic State Cultivation (MSC) protocol, where Clifft's design is most relevant. Building on prior work that was limited to the inject and cultivation stages, we reproduce and extend these results and, for the first time, perform exact end-to-end simulations including the escape stage. This enables a direct comparison between the true $T$-gate circuit and its $S$-proxy approximation at full scale, allowing us to quantify the $T/S$ gap across both intermediate and end-to-end regimes.

    Circuits, analysis scripts, raw collected data, and the exact circuit files are reproducible from the companion \texttt{clifft-paper} repository~\cite{clifftPaperRepo}.

    \subsection{Benchmarks}
    \label{sec:results_qec}

    We compare Clifft against Stim~\cite{gidneyStimFastStabilizer2021} and Tsim~\cite{haenelTsimFastUniversal2026} on three families of fault-tolerant circuits drawn from the literature: pure-Clifford surface-code memory experiments, near-Clifford magic-state preparation circuits, and coherent-noise variants of surface-code memory in which every Clifford gate is followed by a small unitary $R_Z(\theta)$ rotation (citations in Table~\ref{tab:benchmark_throughput} caption). The first family probes graceful degradation in the pure-Clifford limit ($k_{\max}=0$); the second exercises the regime Clifft was designed for ($k_{\max}$ small but nonzero); the third pushes $k_{\max}$ a bit beyond the small limit to see how well Clifft can sustain performance as the active subspace grows.

    Table~\ref{tab:benchmark_throughput} reports sample-time throughput as shots per second. Each circuit is annotated with the peak active dimension $k_{\max}$, the physical qubit count $N$, the unrolled total operation count, and the non-Clifford subset of that count. The \faPlay{} icon next to each circuit name links to the corresponding circuit loaded into Clifft's interactive playground.

    The benchmark circuits use fixed representative noise configurations. The pure-Clifford surface-code circuit is generated with Stim's rotated memory experiment under uniform circuit-level depolarizing noise at $p=10^{-3}$. The coherent-noise variants use the same base circuit-level noise model and add an $R_Z(\theta)$ over-rotation with $\theta=0.02$ co-located with each depolarizing channel. The cultivation circuits are adapted from the Gidney et al.\cite{gidneyMagicStateCultivation2024} MSC templates with stochastic error probabilities scaled to $p=10^{-3}$. Note these are not the end-to-end circuits, but only through the cultivation stage. The distillation circuit is based on one from Refs.~\cite{salesrodriguezExperimentalDemonstrationLogical2025,haenelTsimFastUniversal2026}, with preparation noise $0.05$ and circuit-level depolarizing noise $0.01$. Exact circuit files and generation scripts are included in the companion repository~\cite{clifftPaperRepo}.

    All Clifft and Stim throughput numbers were collected on a single AWS \texttt{c8i.8xlarge} EC2 instance, using 16 physical CPU cores, and Tsim numbers on a single NVIDIA GH200 GPU. Even though these are different compute platforms, they have comparable hourly cloud costs. The reported values are medians over three repetitions. Throughput numbers for Clifft and Stim amortize compile-time overhead, which is milliseconds or less for these circuits. Tsim throughput numbers exclude compilation time and a warmup run. We also note that Tsim throughput can get much faster with smaller physical noise rates $p$; the values in Table~\ref{tab:benchmark_throughput} should therefore be read as fixed-noise benchmark points rather than noise-independent performance.

    \begin{table*}[t]
    \centering
    {\small
    \setlength{\tabcolsep}{4pt}
    \begin{tabular}{l r r r r | r r r}
        \toprule
        \textbf{Circuit} & \textbf{Qubits} & \textbf{Ops} & \textbf{Non-Cliff} & $\mathbf{k_{\max}}$ & \textbf{Clifft} & \textbf{Stim} & \textbf{Tsim} \\
        \midrule
        \multicolumn{8}{l}{\textit{Pure Clifford}\textsuperscript{b}} \\
        Surface code $d=7, r=7$~\clifftplay{surface_d7_r7.stim}    & 118 & 4667 & 0    & 0  & 2.2M            & \textbf{20.1M} & 314.7k \\
        \midrule
        \multicolumn{8}{l}{\textit{Near-Clifford: Magic State}} \\
        Cultivation $d=3$\textsuperscript{c}~\clifftplay{cultivation_d3.stim} & 15 & 676  & 29   & 4  & \textbf{10.4M}  & --             & 27.9k \\
        Cultivation $d=5$\textsuperscript{c}~\clifftplay{cultivation_d5.stim} & 42 & 4379 & 91   & 10 & \textbf{314.4k} & --             & DNC\textsuperscript{a} \\
        Distillation\textsuperscript{d}~\clifftplay{distillation.stim}        & 85 & 1163 & 10   & 5  & \textbf{1.5M}   & --             & \textbf{1.5M} \\
        \midrule
        \multicolumn{8}{l}{\textit{Near-Clifford: Coherent Noise}\textsuperscript{e}} \\
        Surface code $d=3, r=1$~\clifftplay{coherent_d3_r1.stim}              & 26 & 173  & 65   & 5  & \textbf{19.4M}  & --             & 14.4M \\
        Surface code $d=3, r=3$~\clifftplay{coherent_d3_r3.stim}              & 26 & 415  & 195  & 8  & \textbf{1.7M}   & --             & DNC\textsuperscript{a} \\
        Surface code $d=5, r=1$~\clifftplay{coherent_d5_r1.stim}              & 64 & 533  & 209  & 13 & 133.1k          & --             & \textbf{570.9k} \\
        Surface code $d=5, r=5$~\clifftplay{coherent_d5_r5.stim}              & 64 & 2073 & 1045 & 24 & \textbf{0.7}    & --             & DNC\textsuperscript{a} \\
        \bottomrule
    \end{tabular}
    }
    \caption{Sample-time throughput (effective shots/s; median of 3 repetitions) across QEC-relevant circuits. \textbf{Bold} marks the per-row winner. Ops is the total number of operations in the unrolled circuit, and Non-Cliff is the number of non-Clifford operations ($T$ gates or Pauli rotations). $k_{\max}$ is the peak active virtual dimension. The \faPlay{} icon next to each circuit name links to the corresponding circuit loaded into Clifft's interactive playground. Stim results are ``--'' for the non-Clifford circuits it does not support.\\
\textsuperscript{a}DNC: did not compile within a 2 minute time budget.
\textsuperscript{b}Generated via Stim~\cite{gidneyStimFastStabilizer2021}.
\textsuperscript{c}Adapted from~\cite{gidneyMagicStateCultivation2024}.
\textsuperscript{d}Adapted from~\cite{salesrodriguezExperimentalDemonstrationLogical2025,haenelTsimFastUniversal2026}.
\textsuperscript{e}Adapted from~\cite{tuloupComputingLogicalError2026a}.}
    \label{tab:benchmark_throughput}
    \end{table*}

    Clifft is competitive across regimes and strongest in the near-Clifford regime. First, in the pure-Clifford limit, Stim retains a roughly $10\times$ throughput advantage. This is expected: Stim was purpose-built for Clifford simulation and benefits from SIMD parallelism across many shots. Clifft pays an overhead to maintain its frame-factored representation, even when $k_{\max}=0$ and no active state vector is allocated.

    Second, on the near-Clifford magic-state circuits through the cultivation stage, Clifft is substantially faster than Tsim on the tested fixed-noise instances, including a roughly $370\times$ throughput advantage on the $d{=}3$ cultivation benchmark in Table~\ref{tab:benchmark_throughput}\footnote{Specialized pipelines built on top of Tsim can run much faster: Ref.~\cite{wanSimulatingMagicState2026a} reports ${\sim}4{\times}10^6$ shots/s on $d{=}3$ cultivation at $p=5{\times}10^{-4}$.}. Tsim's stabilizer-rank-based strategy is sensitive to the number of non-Clifford gates, whereas Clifft's cost is governed by the peak active dimension. For the cultivation circuits, the non-Clifford structure remains localized, with $k_{\max}=4$ at $d{=}3$ and $k_{\max}=10$ at $d{=}5$. For the $d{=}5$ cultivation circuit, Tsim did not compile within a 2 minute time budget, while Clifft sustains 314k shots/s. In contrast, the distillation circuit has only 10 $T$ gates and $k_{\max}=5$, yielding similar sampling throughput for Clifft and Tsim.

    Third, the coherent-noise sweep makes the $k_{\max}$ dependence concrete. As $k_{\max}$ grows from 5 to 24, Clifft's throughput decreases by several orders of magnitude as the active state vector grows beyond CPU cache size and approaches the dense-state-vector regime. Tsim performs well with one round of measurements, but did not compile within the time budget for the larger repeated-round coherent-noise instances.

    To further explore Clifft's performance as $k_{\max}$ grows, we consider the limit where $k_{\max}=N$. In this limit, Clifft acts as a dense state vector simulator, but leverages the frame-factored representation for efficient evaluation of gates in this regime. Concretely, we sample random noiseless Quantum Volume circuits at various depths $D=N$ and compare Clifft against four general-purpose CPU simulators: Qiskit-Aer~\cite{qiskit2024}, Qulacs\cite{suzukiQulacsFastVersatile2021}, qsim~\cite{quantum_ai_team_and_collaborators_2025_4067237}, and Qrack~\cite{stranoExactApproximateSimulation2023}. This benchmark and simulators were chosen based on a recent study by van Niekerk~et~al.\cite{niekerkComparisonHPCbasedQuantum2024}. These circuits have no near-Clifford structure to exploit, nor any measurement to counteract growth in $k$. They stress-test Clifft's worst-case behavior rather than its design sweet spot. All tests were run on the same 16-core CPU node as above, with default settings for each simulator.

    \begin{figure*}[t]
    \centering
    \includegraphics[width=\textwidth]{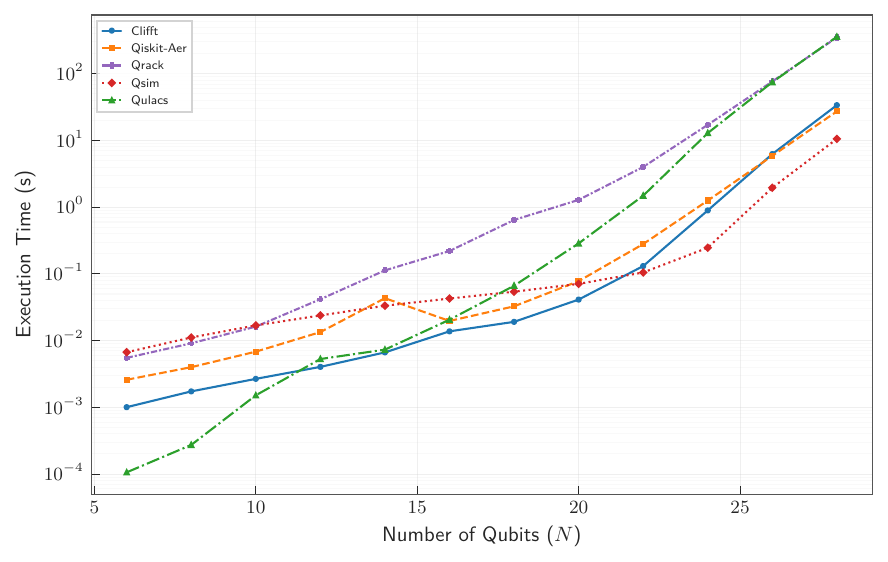}
    \caption{Single-shot execution time vs.\ qubit count $N$ for random Quantum Volume circuits at depth $D=N$, run on a 16-core CPU. This dense-limit benchmark corresponds to the worst case $k_{\max}=N$, where Clifft no longer benefits from near-Clifford structure. Even in this regime, Clifft tracks the leading dense state vector simulators within a small constant factor up to $N=28$; its near-Clifford speedups are shown separately in Table~\ref{tab:benchmark_throughput}.}
    \label{fig:qv_scaling}
    \end{figure*}

    From Fig.~\ref{fig:qv_scaling}, we see Qulacs is fastest at very small circuits through $N=10$, Clifft is fastest from $N=12$ to $N=20$, and qsim takes the lead at $N=22$. Note that qsim uses single-precision floating point, requiring roughly half the memory of Clifft at the same $N$. Clifft remains single-threaded through $N=18$, after which it uses OpenMP to parallelize active-state-vector updates.

    Taken together, these benchmarks show that Clifft spans the regimes targeted by its design. Stim remains substantially faster in the pure-Clifford limit, while Clifft provides its clearest advantage on low-magic fault-tolerant circuits where $N$ is large but $k_{\max}$ remains modest. In the opposite dense limit, where $k_{\max}=N$, Clifft tracks leading CPU state-vector simulators within a small constant factor on the Quantum Volume benchmark. This combination suggests that Clifft is especially useful for non-Clifford circuits with large physical qubit counts and modest, but nontrivial, active dimension.

    \FloatBarrier
    \subsection{Magic State Cultivation}
    \label{sec:results_msc}

    Magic state cultivation~\cite{gidneyMagicStateCultivation2024} prepares high-fidelity logical $|T\rangle$ states by first injecting a physical $T$ state into a small $d=3$ color code, cultivating the state by growing to a $d=5$ color code, and finally escaping it into a full $d=15$ surface code. The protocol is dominated by Clifford structure, but the injection and cultivation stages contain the non-Clifford operations that determine the prepared magic state. The original MSC paper simulated the full protocol using $S$-gate proxy circuits, which replaced all $T$ gates with $S$ gates, and used small-scale state-vector simulations to estimate the degradation introduced by the true $T$-gate circuit. Prior exact simulation work with true $T$-gate circuits has stopped at the $d{=}5$ cultivation stage, using either the GPU-based generalized stabilizer simulator SOFT~\cite{liSOFTHighperformanceSimulator2025} or sparse Pauli-frame representations such as PFSR~\cite{tuloupComputingLogicalError2026a}. The full end-to-end circuit including the escape stage has remained out of reach.

    Throughout the inject and cultivation analysis, we report logical error rates $\epsilon_L$ per kept shot under uniform circuit-level depolarizing noise at physical strength $p$. For the end-to-end analysis, we instead report a conservative infidelity estimate for the prepared $T$ state after decoder correction, as described in Appendix~\ref{sec:t_state_fidelity}. In both cases, we compare against the $S$-gate proxy used as a baseline by Gidney et al.\ in the original cultivation work~\cite{gidneyMagicStateCultivation2024}. We explicitly build on the source code released with that work~\cite{mscGithub}, adding the $T$-gate substitution, measurement changes, and feed-forward corrections needed for the true $T$-gate circuit \cite{clifftPaperRepo}. The discrepancy between the $T$-gate and $S$-proxy estimates---the ``$T/S$ gap''---is the central question Clifft lets us answer at the full end-to-end scale.

    \subsubsection{Inject and Cultivate Stage}
    \label{sec:results_msc_ic}

    Figure~\ref{fig:msc_ic} shows the logical error rate versus expected attempts per kept shot at the cultivation stage for both the $d{=}3$ and $d{=}5$ color codes. We sweep physical noise from $p=5\times 10^{-4}$ to $p=10^{-2}$ at $d{=}3$ and from $p=5\times 10^{-4}$ to $p=2\times 10^{-3}$ at $d{=}5$. The bottom panel of each subplot shows the $T/S$ error ratio with a 95\% Bayesian credible interval, computed as described in Appendix~\ref{app:intervals}. The Clifft $T$-gate curves agree with the SOFT ground-truth values reported in Ref.~\cite{liSOFTHighperformanceSimulator2025} at the noise levels they cover, and are consistent with the PFSR trends reported in Ref.~\cite{tuloupComputingLogicalError2026a}. Both prior works found that the $d=5$ $T/S$ ratio is much closer to $15\times$ than to the $2\times$ extrapolation used in the original MSC paper. At lower physical noise rates, we observe an even larger ratio, reaching approximately $30\times$ at $p=5\times 10^{-4}$ for $d=5$. Although the statistical bands widen at lower $p$, the data suggest that the $T/S$ discrepancy increases over the range studied.

    \begin{figure*}[t]
    \centering
    \begin{subfigure}[b]{0.49\textwidth}
        \centering
        \includegraphics[width=\textwidth]{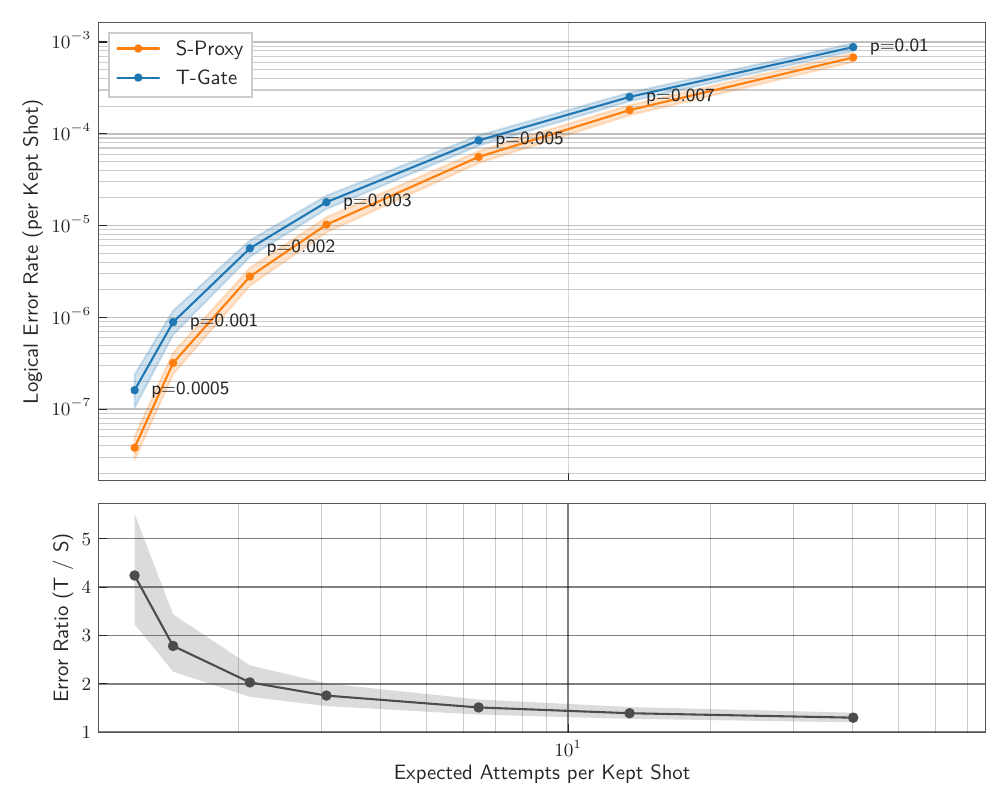}
        \caption{$d=3$ color code.}
        \label{fig:msc_ic_d3}
    \end{subfigure}
    \hfill
    \begin{subfigure}[b]{0.49\textwidth}
        \centering
        \includegraphics[width=\textwidth]{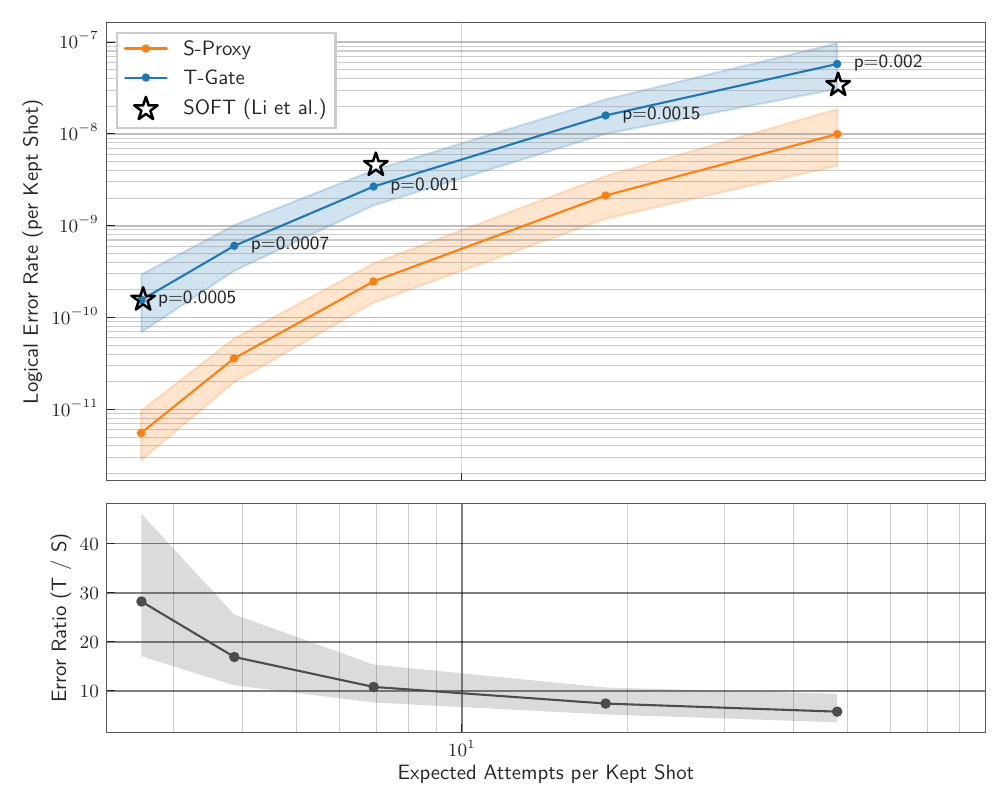}
        \caption{$d=5$ color code.}
        \label{fig:msc_ic_d5}
    \end{subfigure}
    \caption{Inject + cultivate stage: $T$-gate vs $S$-proxy logical error rate per kept shot, estimated using stratified importance sampling across uniform circuit-level depolarizing noise. Shaded regions in the top panels show absolute error-rate values within a factor of 1000 in likelihood relative to the maximum-likelihood estimate, following the convention used by Sinter for rare-event visualization. Shaded regions in the bottom panels show 95\% Bayesian credible intervals for the $T/S$ ratio, computed as described in Appendix~\ref{app:intervals}. For $d=5$, we also show the SOFT results from Table III in Li et al.~\cite{liSOFTHighperformanceSimulator2025} for comparison.}
    \label{fig:msc_ic}
    \end{figure*}

    Beyond the scientific result, Clifft obtains these estimates on a single CPU instance rather than on a GPU cluster. Table~\ref{tab:ic_cost} compares Clifft's true $T$-gate $d{=}5$ inject and cultivation run with SOFT's published $d{=}5$ run. We report wall-clock machine-hours, not summed worker-hours or per-GPU-hours: one machine is a single AWS \texttt{c6i.8xlarge}\footnote{This is an older CPU generation than the one used for the benchmarks in Section~\ref{sec:results_qec}. We chose it for lower spot-instance cost on these longer runs.} instance with 16 physical CPU cores for Clifft, and the full 16-H800 GPU cluster used by SOFT~\cite{liSOFTHighperformanceSimulator2025}. On this basis, Clifft reaches comparable low-rate estimates in approximately 12 machine-hours, versus approximately 388 machine-hours for SOFT, a $\sim 32\times$ reduction in machine time before accounting for the substantial hourly cost difference between a CPU instance and a 16-GPU cluster.

    This reduction comes from two independent factors. First, Clifft uses stratified importance sampling, following Ref.~\cite{tuloupComputingLogicalError2026a}, which reduces the total number of required shots by $\sim 2.5\times$. Second, Clifft is faster per shot: simulating one shot of this circuit takes $\sim 7.4\times 10^{-6}$ s on one c6i.8xlarge core, versus $\sim 9.4\times 10^{-5}$ s on one H800 GPU as reported in Table V of Li et al.~\cite{liSOFTHighperformanceSimulator2025}, a $\sim 13\times$ single-shot speedup. Because both runs use shot-level data parallelism across 16 workers or GPUs, this per-shot speedup carries through directly to per-machine throughput.

    \begin{table*}[t]
    \centering
    \adjustbox{max width=\textwidth}{%
    \small
    \begin{tabular}{l l r r r r r r}
        \toprule
        \textbf{Simulator} & \textbf{Hardware} & \textbf{Total shots} & \textbf{Shots/s} & \textbf{Machine-h} & $\boldsymbol{\epsilon_L}$($p{=}5{\times}10^{-4}$) & $\boldsymbol{\epsilon_L}$($p{=}10^{-3}$) & $\boldsymbol{\epsilon_L}$($p{=}2{\times}10^{-3}$) \\
        \midrule
        Clifft ($T$-gate)  & c6i.8xlarge (CPU)        & $9.6{\times}10^{10}$ & 2.17M & 12   & $1.56{\times}10^{-10}$ & $2.67{\times}10^{-9}$  & $5.78{\times}10^{-8}$ \\
        SOFT ($T$-gate)    & $16{\times}$ H800 (GPU)  & $2.4{\times}10^{11}$ & 171k  & 388  & $1.57{\times}10^{-10}$ & $4.59{\times}10^{-9}$  & $3.41{\times}10^{-8}$ \\
        \bottomrule
    \end{tabular}
    }
    \caption{Cost comparison for the true $T$-gate $d{=}5$ inject and cultivation stage. SOFT total shots and per-$p$ error rates are reproduced from Tables III and V of Li et al.~\cite{liSOFTHighperformanceSimulator2025}. Reported throughput (Shots/s) and runtime (Machine-h) are normalized to the full machine configuration used for each simulator: one \texttt{c6i.8xlarge} instance with 16 physical CPU cores for Clifft, and the full 16-H800 GPU cluster for SOFT. SOFT's machine-level shots/s is computed by multiplying the reported single-H800 rate ($\sim 10.7$k shots/s, the inverse of the $9.37 \times 10^{-5}$ s/shot single-GPU runtime in Li et al.\ Table V) by 16 (the cluster's GPU count); SOFT's machine-hours come from dividing total shots by this aggregate rate. $\epsilon_L$ is the logical error rate per kept shot at the listed physical noise strength.}
    \label{tab:ic_cost}
    \end{table*}

    \FloatBarrier
    \subsubsection{End-to-End Cultivation Including the Escape Stage}
    \label{sec:results_msc_e2e}

    Beyond the cultivation stage, the full MSC protocol grows the cultivated state from the small color code into a larger surface code of distance $d{=}15$. This end-to-end circuit has 463 qubits, but the escape stage itself is purely Clifford. Unlike the inject and cultivation stages, which rely on detector post-selection to reject failed preparations, the escape stage relies on a decoder to identify and correct errors. The decoder reports a gap between the most likely and second most likely error configurations, and shots with a gap below a chosen threshold are rejected. Sweeping this gap threshold trades lower acceptance probability for improved output quality.

    Because we use the decoder constructed for the $S$-proxy circuit, the plots below report a conservative infidelity estimate for the prepared $T$ state after decoder correction rather than a direct logical error rate. Appendix~\ref{sec:t_state_fidelity} explains this estimator. As before, we use the code released alongside the original MSC paper for decoding, gap thresholding, and plotting, with our modifications in \cite{clifftPaperRepo}.

    \begin{figure*}[t]
    \centering
    \begin{subfigure}[b]{0.49\textwidth}
        \centering
        \includegraphics[width=\textwidth]{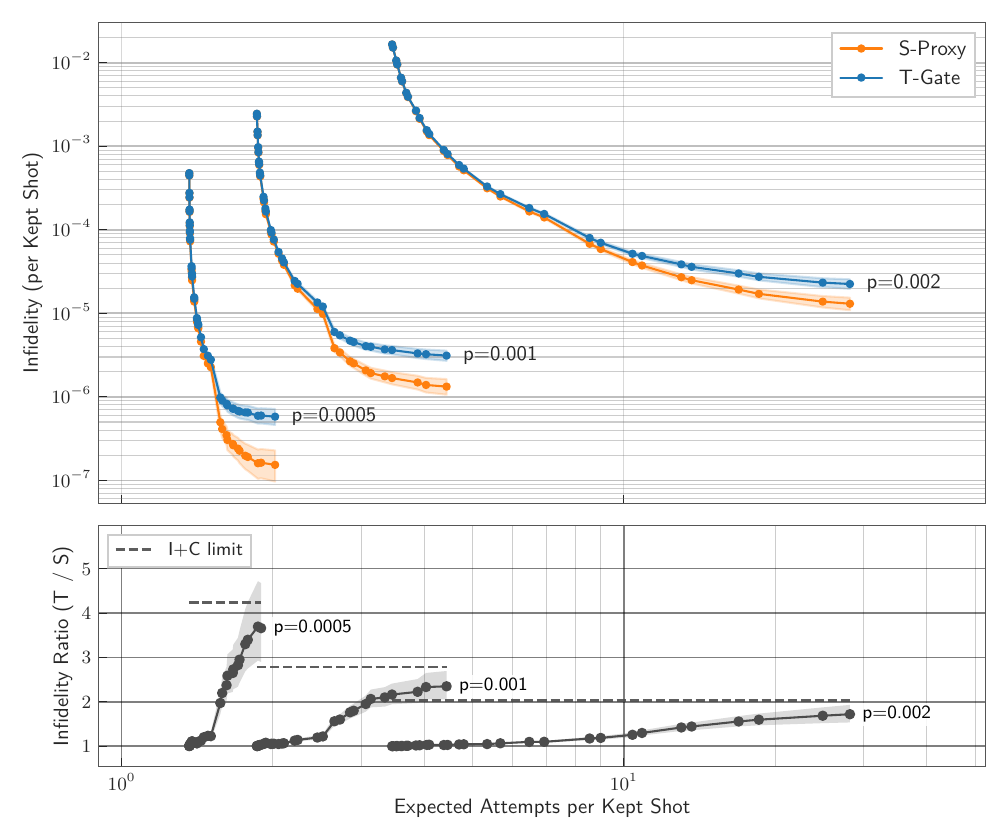}
        \caption{End-to-end: $d=3$ cultivation, $d=15$ escape}
        \label{fig:msc_e2e_d3}
    \end{subfigure}
    \hfill
    \begin{subfigure}[b]{0.49\textwidth}
        \centering
        \includegraphics[width=\textwidth]{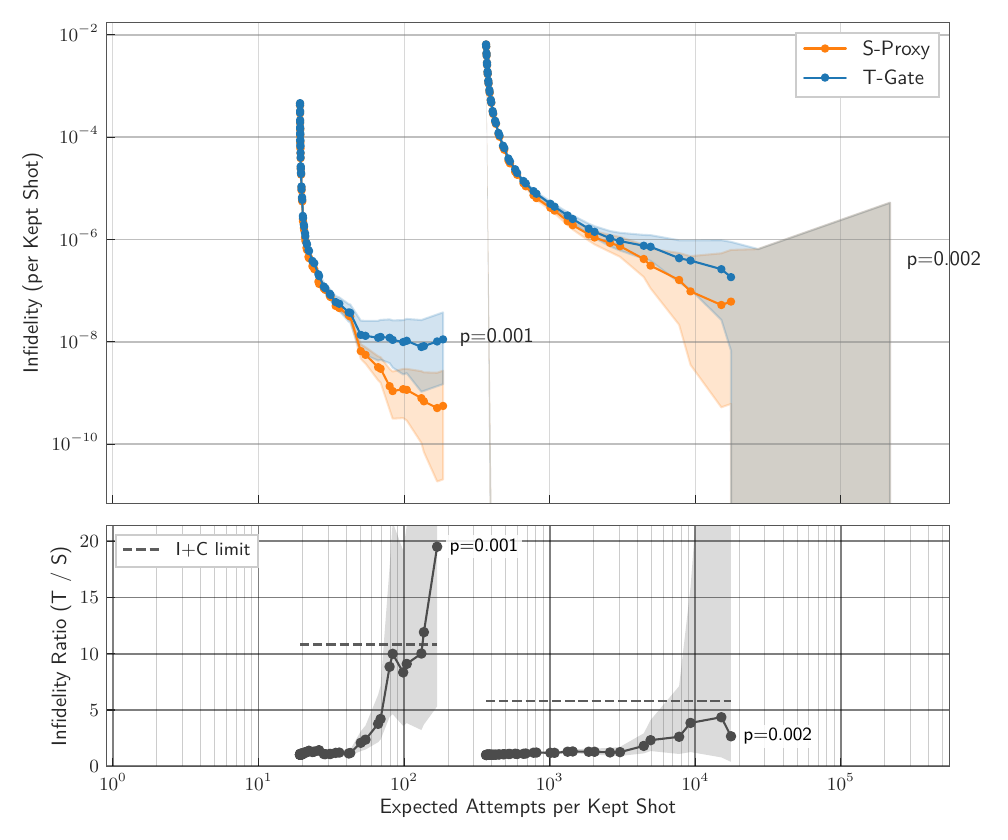}
        \caption{End-to-end: $d=5$ cultivation, $d=15$ escape}
        \label{fig:msc_e2e_d5}
    \end{subfigure}
    \caption{End-to-end magic state cultivation including the escape stage to a $d=15$ surface code. Top panels: conservative infidelity estimate per kept shot after decoder correction for the true $T$-gate circuit (Clifft) and the $S$-proxy baseline (Gidney et al.~\cite{gidneyMagicStateCultivation2024}). Bottom panels: $T/S$ infidelity ratio with 95\% Bayesian credible intervals; the dashed line marks the inject and cultivation stage ratio from Section~\ref{sec:results_msc_ic} as a reference. Note that the credible intervals for $d=5$ cultivation extend beyond the plotted range due to low rare-event statistics.}
    \label{fig:msc_e2e}
    \end{figure*}

    Figure~\ref{fig:msc_e2e} shows the result of direct end-to-end simulations with $d=3$ and $d=5$ cultivate stages, overlaying the Clifft $T$-gate curves with the $S$-proxy data released by Gidney et al. These end-to-end runs did not use importance sampling. We collected 3 billion $T$-gate shots at $d=3$ ($\sim 18$ machine-hours) and 390 billion shots at $d=5$ ($\sim 335$ machine-hours), both normalized to a single AWS \texttt{c6i.8xlarge} instance using the same machine-scale convention as in Table~\ref{tab:ic_cost}. Even at this scale, the deep decoder-gap tail of the $d=5$ desaturation curve remains shot-starved, which appears as widening 95\% credible-interval bands in the ratio plot.

    Overall, these results clarify how the $T/S$ discrepancy manifests at the full protocol level. At low decoder-gap thresholds, the conservative end-to-end infidelity estimate is dominated by escape-stage decoding failures, which are largely insensitive to the underlying $T$ versus $S$ circuit differences. The two curves therefore overlap in this regime. As the gap threshold increases and decoder-induced failures are progressively filtered out, the observed $T/S$ ratio moves toward the inject and cultivation stage ratio, where the intrinsic discrepancy between the true $T$ circuit and the $S$ proxy is pronounced. This separation of regimes indicates that the cultivation-stage discrepancy is a strong predictor of the high-confidence end-to-end behavior, while the escape-stage decoder controls how quickly that asymptotic regime is reached as a function of gap-threshold post-selection.

    \FloatBarrier
    \section{Conclusion}
    \label{sec:conclusion}

    We introduced Clifft, an exact simulator for near-Clifford quantum circuits. Clifft factors the simulated state into an offline Clifford frame, an online Pauli frame, and a dynamically sized active state vector. This shifts deterministic Clifford-coordinate tracking to compile time and confines sample-time dense evolution to the active virtual dimension $k$. The resulting compile-once, sample-many architecture extends the practical workflow established by Stim into non-Clifford regimes while retaining exact treatment of noise, mid-circuit measurement, and classical control.

    Across benchmarks, Clifft remained competitive in the dense-state limit, degraded gracefully in the pure-Clifford limit, and showed its clearest advantage on low-magic fault-tolerant circuits where $k_{\max}$ remains modest despite large physical system size. This enabled exact end-to-end simulation of magic state cultivation through the escape stage, over hundreds of billions of shots on commodity CPUs. These simulations show that the discrepancy between the true $T$-gate circuit and the $S$-proxy approximation persists at the full-protocol level, becoming visible once decoder-gap post-selection suppresses escape-stage-dominated failures.

    Several directions follow from this work. Better global scheduling and localization heuristics may reduce $k_{\max}$ or the total time spent at large active dimension beyond what the current greedy Pauli localization and HIR optimization passes achieve. Broader gate support, including other non-Clifford resources, would expand the range of fault-tolerant protocols Clifft can target. Coherent-noise workloads motivate approximation schemes, hybrid stabilizer-rank methods, and richer support for non-Pauli noise models such as leakage. At the systems level, improved memory layout, lower dispatch overhead, GPU execution, and more aggressive vectorization of active-array kernels provide direct paths to higher throughput.

    More broadly, HIR provides a compiler-facing representation of early fault-tolerant circuits, not only a simulator-internal IR. In Clifft, HIR is used to expose non-Clifford structure, commute operations, and specialize circuits for the SVM backend. The same representation could support stronger Clifford+$T$ optimization passes, integration with Pauli-based compiler tools, and new backends targeting fault-tolerant instruction sets or architecture-specific primitives. This suggests a path from Clifft as a simulator toward a broader compiler toolchain for early fault-tolerant quantum programs. Taken together, these results show that frame-factored exact simulation occupies a useful middle ground between pure stabilizer simulation and dense state-vector methods, providing an exact and practical path to simulating large fault-tolerant circuits with substantial Clifford structure and localized non-Clifford effects.

    \section*{Acknowledgements}
    We thank Will Zeng, Nathan Shammah and the rest of the Unitary Foundation team for their support and feedback throughout this project. We also thank Craig Gidney for feedback on generating the $T$-gate MSC circuits, and for his commitment to open sourcing Stim and the code and data for the original MSC paper, which were invaluable for this work. We also thank the developers of SOFT and TSim for open sourcing their code and data, which enabled direct comparison against their results.

    We used generative AI tools, including Gemini 3, ChatGPT+Codex 4.4/4.5, and Claude Opus 4.6/4.7, during the research, software-development, and writing workflow for this project. These tools assisted with code generation and review, implementation analysis, manuscript editing, and checks of selected derivations or arguments.

    This work was supported by the U.S. Department of Energy, Office of Science, Office of Advanced Scientific Computing Research, Accelerated Research in Quantum Computing under Award Number DE-SC0025336. This material is also based upon work supported by the U.S. Department of Energy, Office of Science, National Quantum Information Science Research Centers, Quantum Science Center.

    \bibliographystyle{quantum}
    \bibliography{refs}

    \appendix

    \section{Constructive Proof of Pauli Localization}
    \label{app:pauli_localization}
    We provide below a constructive proof of Lemma~\ref{lem:pauli_localization}, detailing the $\mathcal{O}(N)$ greedy algorithm for localizing any multi-qubit virtual Pauli operator to a single virtual qubit.
    \begin{proof}
        Let the multi-qubit virtual Pauli be parameterized by its Boolean support vectors $\mathbf{x}, \mathbf{z} \in \{0,1\}^N$ such that $\tilde{P} \propto X^{\mathbf{x}} Z^{\mathbf{z}}$. We construct $V$ in two mutually exclusive cases:

        \textbf{Case 1 ($\mathbf{x} \neq 0$):} The operator contains $X$-basis support. Select a pivot virtual qubit $v$ such that $x_v = 1$. For every other qubit $q \neq v$ where $x_q = 1$, append $\text{CNOT}_{v \to q}$ to $V$, which conjugates as:
        \begin{equation}
            \text{CNOT}_{v \to q} (X_v \otimes X_q) \text{CNOT}_{v \to q}^\dagger = X_v \otimes I_q
        \end{equation}
        This annihilates the $X$-support on $q$. Any $Z$-support on $v$ or $q$ propagates to a new vector $\mathbf{z}'$, reducing the intermediate operator to $\tilde{P}' \propto X_v \otimes Z^{\mathbf{z}'}$.

        To eliminate the remaining $Z$-support, append $\text{CZ}_{v, q}$ to $V$ for every $q \neq v$ where $z'_q = 1$. The CZ conjugation rule gives:
        \begin{equation}
            \text{CZ}_{v, q} (X_v \otimes Z_q) \text{CZ}_{v, q}^\dagger = X_v \otimes I_q
        \end{equation}
        This annihilates the $Z$-support on all $q \neq v$ without altering the $X_v$ pivot, fully localizing the operator to $v$ as either $X_v$ or $Y_v$. If $Y_v$ results, a virtual Phase gate ($S_v Y_v S_v^\dagger = -X_v$) completes the compression to $X_v$.

        \textbf{Case 2 ($\mathbf{x} = 0$):} The operator is a pure $Z$-string ($\mathbf{z} \neq 0$). Select a pivot $v$ such that $z_v = 1$. For every $q \neq v$ where $z_q = 1$, append $\text{CNOT}_{q \to v}$ (control $q$, target $v$) to $V$. Under conjugation:
        \begin{equation}
            \text{CNOT}_{q \to v} (Z_q \otimes Z_v) \text{CNOT}_{q \to v}^\dagger  = I_q \otimes Z_v
        \end{equation}
        This annihilates the $Z$-support on all $q \neq v$. The operator is successfully compressed to $Z_v$.

        Each two-qubit gate monotonically erases generator support on $q$ without spreading to previously cleared qubits, so the sequence $V$ requires at most $\mathcal{O}(N)$ gates and guarantees isolation to a single virtual qubit.
    \end{proof}

    \section{Active Operations in the Frame-Factored Representation}
    \label{app:active_operations}

    This section details how active operations are evaluated within the frame-factored representation.

    \subsection{Pauli Rotations}
    \label{sec:continuous_rotations}

    Consider a continuous non-Clifford rotation $\exp(-i \theta P_q)$ applied in the physical frame (e.g., the $T$ gate where $P_q = Z_q$ and $\theta = \pi/8$). The Heisenberg mapping transforms the physical generator $P_q$ into a multi-qubit virtual string $\tilde{P}_q$.

    Applying the virtual localization sequence $V$ yields a single-qubit Pauli $P_v$. Because $V$ is unitary, it distributes into the exponential:
    \begin{equation}
    	V \exp\left(-i \theta \tilde{P}_q\right) V^\dagger = \exp\left(-i \theta V \tilde{P}_q V^\dagger\right) = \exp\left(-i \theta \alpha P_v\right)
    \end{equation}
    If localization maps the generator natively to $X_v$ or $Y_v$, a virtual basis change (e.g., $H_v$) is appended to $V$ to ensure the rotation evaluates around the $Z$-axis (with $\alpha$ absorbed by modifying $\theta$ if necessary).

    To evaluate this localized rotation against the active state vector, it must first commute past the virtual Pauli frame $\tilde{P}^{(t)}$. Because Pauli matrices either commute or anti-commute, passing the rotation through the frame extracts a deterministic algebraic sign:
    \begin{equation}
    	\exp\left(-i \theta Z_v\right) \tilde{P}^{(t)} = \tilde{P}^{(t)} \exp\left(-i \theta (-1)^c Z_v\right)
    \end{equation}
    where $c=1$ if $\tilde{P}^{(t)}$ anti-commutes with $Z_v$ (e.g., possesses $X_v$ or $Y_v$ parity), and $c=0$ otherwise. Thus, the effective continuous phase applied to the subspace depends strictly on the boolean parity of the Pauli frame. The localized rotation acts on the virtual state vector in one of three ways:

    \textbf{1. Trivial Phase ($v \in D$, generator natively $Z_v$):}
    Because $v \in D$, the state on axis $v$ is $|0\rangle_v$. The diagonal rotation acts simply as a scalar multiplier $e^{-i\theta(-1)^c}$, which is absorbed directly into the global scalar $\gamma$. The continuous vector $|\phi\rangle_A$ remains isolated, and the active dimension $k$ is unchanged.

    \textbf{2. Subspace Expansion ($v \in D$, requires $H_v$):}
    If the generator localized to the $X$-basis, the necessary virtual Hadamard maps $|0\rangle_v \to |+\rangle_v$. This operation promotes $v$ to the active set ($A \leftarrow A \cup \{v\}$, $k \leftarrow k+1$). The state vector dimensionality expands via the tensor product $|\phi\rangle_A \leftarrow |\phi\rangle_A \otimes |+\rangle$, and the phase rotation is subsequently evaluated on the newly active degree of freedom.

    \textbf{3. Active Subspace Rotation ($v \in A$):}
    If the virtual qubit is already part of the non-Clifford superposition, active set is unchanged. The diagonal phase is applied directly to the corresponding tensor factor within the complex amplitudes of $|\phi\rangle_A$.

    \subsection{Projective Pauli Measurements}
    \label{app:measurement_updates}

    After Heisenberg mapping and Pauli localization, a physical measurement is represented by a single-qubit virtual observable $M_v \in \{X_v,Z_v\}$ and projector
    \begin{equation}
        \Pi_m = \frac{1}{2}\bigl(I + (-1)^m M_v\bigr),
    \end{equation}
    where $m \in \{0,1\}$ is the physical measurement outcome.

    Before the projector acts on the virtual state, it is commuted past the virtual Pauli frame:
    \begin{equation}
        \Pi_m \tilde{P}^{(t)}
        =
        \tilde{P}^{(t)} \Pi_{m \oplus p},
    \end{equation}
    where $p=1$ if $M_v$ anticommutes with $\tilde{P}^{(t)}$, and $p=0$ otherwise. The Pauli frame therefore contributes a deterministic inversion of the measurement branch.

    Applying the shifted projector $\Pi_{m \oplus p}$ to
    $|\phi\rangle_A \otimes |0\rangle_D$ yields two cases.

    \textbf{Dormant measurement ($v \in D$).}
    The virtual qubit $v$ is known to be in the state $|0\rangle_v$.
    If $M_v=Z_v$, the measurement is deterministic in the virtual basis. The physical outcome is fixed by the Pauli-frame shift, $m=p$, and no state-vector update is required.

    If $M_v=X_v$, the measurement is conjugate to the dormant $|0\rangle_v$ basis. The physical outcome is uniformly random, again shifted by the Pauli frame. The post-measurement virtual state is an $X_v$ eigenstate. To restore the dormant convention, a virtual Hadamard $H_v$ is absorbed into the Clifford frame and Pauli frame:
    \begin{equation}
        U_C \leftarrow U_C H_v,
        \qquad
        \tilde{P} \leftarrow H_v \tilde{P} H_v.
    \end{equation}
    The active state vector is unchanged and $v$ remains dormant.

    \textbf{Active measurement ($v \in A$).}
    If $v$ belongs to the active set, the shifted projector is applied directly to the corresponding tensor factor of $|\phi\rangle_A$. For a $Z_v$ measurement, this masks amplitudes according to the measured branch. For an $X_v$ measurement, it folds amplitudes across the measured axis. After normalization, the measured qubit is disentangled from the remaining active state. If needed, a virtual basis change such as $H_v$ is absorbed into the frames so the decoupled qubit returns to the dormant $|0\rangle_v$ convention. The qubit is then demoted from $A$ to $D$, reducing the active dimension by one.
    \subsection{Conditional Pauli Operations}
    \label{sec:conditional_paulis}

    Unlike continuous rotations, feed-forward corrections and stochastic noise events do not perturb amplitudes and therefore bypass Pauli localization entirely. They are modeled as discrete conditional Pauli operations governed by a classical boolean parameter $c \in \{0, 1\}$—representing either a deterministic classical measurement outcome or a random sample drawn from an error distribution.

    We Heisenberg map the physical Pauli operator $E$ into a multi-qubit virtual Pauli string $\tilde{E}$:
    \begin{equation}
    	|\psi^{(t+1)}\rangle = E^c |\psi^{(t)}\rangle = \gamma^{(t)} U_C^{(t)} (\tilde{E})^c \tilde{P}^{(t)} \Big( |\phi\rangle_A \otimes |0\rangle_D \Big)
    \end{equation}
    Because the $N$-qubit Pauli group is closed under multiplication, the product $(\tilde{E})^c \tilde{P}^{(t)}$ collapses into a new phase-free virtual frame $\tilde{P}^{(t+1)}$ and an overall scalar phase $\alpha \in \{\pm 1, \pm i\}$:
    \begin{equation}
    	|\psi^{(t+1)}\rangle = \Big( \gamma^{(t)} \cdot \alpha \Big) U_C^{(t)} \tilde{P}^{(t+1)} \Big( |\phi\rangle_A \otimes |0\rangle_D \Big)
    \end{equation}
    The resulting state update is executed entirely within the discrete basis tracking frames, leaving the active state vector $|\phi\rangle_A$ unchanged.

    \section{Measuring the $T$ state fidelity on the surface code}
    \label{sec:t_state_fidelity}

    Directly measuring the encoded $|T\rangle$ state in the $\frac{1}{\sqrt{2}}(X+Y)$ basis is not available as a standard transversal measurement in the surface code. The Eastin-Knill theorem restricts transversal logical gates to a finite group \cite{eastinRestrictionsTransversalEncoded2009}. Instead, we use a Clifft-native expectation-value probe to evaluate the logical $Y_L$ expectation value after decoder correction. This gives a conservative estimate of the output $T$-state fidelity in the regime studied here.

    For an ideal $|T\rangle$ state,
    \begin{equation}
        \langle X \rangle = \frac{1}{\sqrt{2}},
        \qquad
        \langle Y \rangle = \frac{1}{\sqrt{2}},
        \qquad
        \langle Z \rangle = 0 .
    \end{equation}
    The fidelity of a prepared logical state $\rho$ relative to the ideal logical $|T\rangle$ state is therefore
    \begin{equation}
        \label{eq:t_state_fidelity}
        F
        =
        \frac{1}{2}
        +
        \frac{1}{2\sqrt{2}}\langle X_L \rangle
        +
        \frac{1}{2\sqrt{2}}\langle Y_L \rangle .
    \end{equation}
    If $\langle X_L\rangle \geq \langle Y_L\rangle$, then replacing $\langle X_L\rangle$ by $\langle Y_L\rangle$ gives the conservative estimate
    \begin{equation}
        \label{eq:t_state_fidelity_bound}
        F
        \geq
        \frac{1}{2}
        +
        \frac{1}{\sqrt{2}}\langle Y_L \rangle .
    \end{equation}
    Equivalently, the reported infidelity
    \begin{equation}
        1 -
        \left(
            \frac{1}{2}
            +
            \frac{1}{\sqrt{2}}\langle Y_L \rangle
        \right)
    \end{equation}
    is a conservative upper bound on the true infidelity when this inequality holds.

    This inequality is expected for the surface-code escape stage under the symmetric circuit-level noise and decoder used here. After syndrome extraction and decoding, residual logical failures can be described by an effective logical Pauli channel,
    \begin{equation}
        \Lambda(\rho)
        =
        (1 - P(X_L) - P(Y_L) - P(Z_L))\rho
        + P(X_L) X_L\rho X_L
        + P(Y_L) Y_L\rho Y_L
        + P(Z_L) Z_L\rho Z_L .
    \end{equation}
    Logical errors that anticommute with a measured observable attenuate its expectation value:
    \begin{align}
        \label{eq:logical_attenuation}
        \langle X_L \rangle_{\mathrm{noisy}}
        &=
        \bigl(1 - 2P(Y_L) - 2P(Z_L)\bigr)
        \langle X_L \rangle_{\mathrm{ideal}},
        \\
        \langle Y_L \rangle_{\mathrm{noisy}}
        &=
        \bigl(1 - 2P(X_L) - 2P(Z_L)\bigr)
        \langle Y_L \rangle_{\mathrm{ideal}} .
    \end{align}
    In a CSS surface code, logical $X_L$ and $Z_L$ failures are produced by single homologically nontrivial error chains of the corresponding type. A logical $Y_L$ failure requires both logical components to occur together, i.e., two independent orthogonal failure chains in the same local patch. Thus $P(Y_L)$ is expected to be much smaller than $P(X_L)$ and $P(Z_L)$, and $\langle Y_L\rangle$ is attenuated at least as strongly as $\langle X_L\rangle$.

    We verify this expectation for the setting used in the end-to-end MSC simulations by running surface-code memory experiments in the logical bases $b \in \{X,Y,Z\}$. For a $d=15$ surface code under uniform circuit-level depolarizing noise at $p=0.003$, sampled until $10^5$ observed errors, we find
    \begin{align}
        P(X_L) &= 4.679 \times 10^{-4}, \nonumber \\
        P(Z_L) &= 4.702 \times 10^{-4}, \nonumber \\
        P(Y_L) &= 8.372 \times 10^{-6}. \nonumber
    \end{align}
    These results confirm the expected hierarchy $P(Y_L) \ll P(X_L), P(Z_L)$ for the studied decoder and noise model. Consequently, using the decoder-corrected logical $Y_L$ expectation value gives a conservative fidelity estimate for the escaped $T$ states reported in Section~\ref{sec:results_msc_e2e}.

	\section{Bayesian Credible Intervals for Rate Ratios}
    \label{app:intervals}

    The ratio bands in Section~\ref{sec:results_msc} summarize uncertainty in the relative error or infidelity ratio between the true $T$-gate circuit and the $S$-proxy circuit. For these ratio estimates, we use a simple Bayesian model for the ratio of two unknown event probabilities.

    Consider two independent estimates with observed event counts $k_1,k_2$ and trial counts $n_1,n_2$. Each unknown rate $p_i$ is modeled with a binomial likelihood and Jeffreys prior \cite{jeffreys1946invariant, lee2012bayesian},
    \begin{equation}
        p_i \sim \mathrm{Beta}\left(\frac{1}{2},\frac{1}{2}\right).
    \end{equation}
    After observing $k_i$ events in $n_i$ trials, conjugacy gives the posterior
    \begin{equation}
        p_i \mid k_i,n_i
        \sim
        \mathrm{Beta}
        \left(
            k_i + \frac{1}{2},
            n_i - k_i + \frac{1}{2}
        \right).
    \end{equation}

    The posterior distribution of the ratio $r=p_1/p_2$ does not have a simple closed form, so we estimate it by Monte Carlo sampling. We draw paired samples from the two Beta posteriors and compute their elementwise ratios. The median of the resulting ratio samples is used as the point estimate, and the 2.5th and 97.5th percentiles define the reported 95\% Bayesian credible interval. Under this model and prior, the posterior probability that the true ratio lies in this interval is 95\%.

    This appendix describes only the ratio intervals shown in the lower panels of Figures~\ref{fig:msc_ic} and~\ref{fig:msc_e2e}. The absolute error-rate bands in the upper panels of Figure~\ref{fig:msc_ic} use a separate likelihood-support convention for rare-event visualization, as described in that figure caption and implemented in Stim/Sinter \cite{gidneyStimFastStabilizer2021}.

\end{document}